\documentclass[review, authoryear]{article}

\usepackage{latexsym,ifthen,amssymb}
\usepackage{natbib}
\usepackage[toc,page,title,titletoc,header]{appendix}
\usepackage{multirow}
 \usepackage{longtable}
 \usepackage{rotating}
 \usepackage{graphicx, amsmath, amsthm, amssymb,float,setspace,color, multirow}
\usepackage{enumitem}

\title{On the Log Partition Function of Ising Model
on Stochastic Block Model}

\author{\large
\textsc{Lu Liu\footnote{ g.jiayi.liu@gmail.com}}
\\ \textit{Central South University}
 }

\date{}

\def\mbbP{\mathbb{P}}

\def\h{\hat}
\def\oE{\operatorname{\mathbb{E}}}

\newtheorem{theorem}{Theorem}[section]
\newtheorem{proposition}[theorem]{Proposition}
\newtheorem{lemma}[theorem]{Lemma}
\newtheorem{corollary}[theorem]{Corollary}

\theoremstyle{definition}
\newtheorem{definition}[theorem]{Definition}

\newtheorem{condition}[theorem]{Condition}

\theoremstyle{remark}
\newtheorem{remark}[theorem]{Remark}

\begin{document}
\maketitle

\begin{abstract}

A sparse stochastic block model (SBM) with two communities is defined by the community probability $\pi_0,\pi_1$, and the connection probability between communities $a,b\in\{0,1\}$, namely $q_{ab} = \frac{\alpha_{ab}}{n}$. When $q_{ab}$ is constant in $a,b$, the random graph is simply the Erd\H{o}s-R\'{e}ny random graph. We evaluate the log partition function of the Ising model on sparse SBM with two communities.

As an application, we give consistent parameter estimation of the sparse SBM with two communities in a special case. More specifically, let $d_0,d_1$ be  the average degree of  the two communities, i.e.,  $d_0\overset{def}{=}\pi_0\alpha_{00}+\pi_1\alpha_{01},d_1\overset{def}{=}\pi_0\alpha_{10}+\pi_1\alpha_{11}$.
We focus on the regime $d_0=d_1$ (the regime $d_0\ne d_1$ is trivial). In this regime, there exists $d,\lambda$ and $r\geq 0$ with $\pi_0=\frac{1}{1+r}, \pi_1=\frac{r}{1+r}$, $\alpha_{00}=d(1+r\lambda), \alpha_{01}=\alpha_{10} = d(1-\lambda), \alpha_{11} = d(1+\frac{\lambda}{r})$. We give a consistent estimator of $r$ when $\lambda<0$. The estimator of $\lambda$ given by \citep{mossel2015reconstruction}  is  valid in  the general situation. We also provide a random clustering algorithm which does not require knowledge of parameters and which is positively correlated with the true community label when $\lambda<0$.


\end{abstract}

\section{Introduction}

Stochastic block model (SBM),
also known as planted partition
model, is one of the most
commonly used generative network model.
 In this model, every
node $i\in V=\{1,2,\cdots,n\}$ is assigned a latent type
(community \emph{label})
$\sigma_i$ with probability
$\pi_{\sigma_i}$.
Conditioned on
node types,
the connection between
nodes is independent of each other.
For every two nodes $i$ and $j$,
 the conditional connection probability
is
$q_{\sigma_i\sigma_j}$,
which
depends  on the types of the two nodes.
 Denote the
SBM defined by $q,\pi$ as $SBM(q,\pi)$.
When the connection probability
 is a constant,
  the model becomes
   the Erd\H{o}s-R\'{e}ny model $\mathcal{G}(n,q)$.
 The \emph{clustering (community detection)}
 problem is to infer the  latent
types from the network structure.
This is an important problem
in many areas such as computer science,
social network analysis,
statistics, machine learning, biology and
image processing
(see \citep{Fortunato2010Community} for a
thorough introduction).
The \emph{parameter estimation} problem
is to estimate model parameters
$\pi_a,q_{ab}$.

SBM is one of the most popular network model,
 not only because of its simplicity,
 but also for the following reasons.
First, it well fits a lot of 
real 
world data in the following fields,
 social network
 \citep{holland1981exponential,Newman2002Random,Robins2009Recent}
 (notably, \cite{holland1981exponential}
 first proposed SBM),
 biology \citep{rohe2011spectral},
gene regulatory network
\citep{Schlitt2007Current,Pritchard2000Inference},
image processing
\citep{Shi2000Normalized,Sonka2008Image}.
Second,
the model
is a nice tool to investigate
clustering algorithms from
the theoretical
perspective.
 Some early works in this stream are
 \citep{Dyer1989The,Jerrum1998The,Condon2001Algorithms}.
 Their
 focus is the algorithmic aspects of
  the min-bisection problem.
Later a vast amount of research is carried out
to study and compare the performance of various
clustering
algorithms on SBM. Roughly speaking,
these algorithms
can be divided into the following categories.
Modularity algorithm
\citep{Newman2004Finding},
likelihood algorithm
\citep{bickel2009nonparametric,Choi2012Stochastic,amini2013pseudo,Celisse2012Consistency} etc.,
and most importantly,
spectral algorithm
\citep{chatterjee2011random,balakrishnan2011noise,Jin2014Fast,Sarkar2013Role,Krzakala2013Spectral} etc..

Notably, \cite{bickel2009nonparametric} provided
a general framework to establish
consistency of clustering.
It was further extended by \citep{zhao2012consistency}
to establish consistency of many
clustering algorithms
in more general models.
 These algorithms include
maximum likelihood estimation and
various modularity methods.
The technique is largely based on finite covering
plus concentration inequality.
This line is also followed to establish consistency
of spectral clustering \citep{Lei2015Consistency}.
Although there need the
 evaluation of the norm of a random matrix,
 which is more complicated.

\subsection{Related work on sparse stochastic block model}
\label{sbmsubsec0}

In reality, many networks are  sparse.
For example,
 \cite{Leskovec2008Statistical}.
 found that many   
large networks with
millions of nodes have an average
degree less than $20$.
These networks include,
social
networks like LinkedIn and MSN Messenger;
collaboration networks in
movies and on the arXiv
(see also \citep{Strogatz2001Exploring});
and some biological networks.

Despite the vast amount of literature on SBM,
most of the literatures has focused on dense SBM.
Where dense means that the average degree
scales with network size and is usually
of order at least $\log n$.
However, very few is known for sparse SBM.
A sparse SBM refers to
the SBM
with constant level  degree, i.e.,
$q_{ab} = \frac{\alpha_{ab}}{n}$.
Sparse SBM is generally more difficult to
handle. For instance,
in contrast with dense SBM, consistent
clustering is impossible since
there exists a constant portion of
isolated nodes, and there is no
way to identify the community label of
an isolated node. Also note that
the local structure of the network
can not be distinguished from that of
a Erd\H{o}s-R\'{e}ny model
$\mathcal{G}(n,\frac{d}{n})$ if the expected
degree of each node is $d$. 
For instance,
in such SBM, the degree of the nodes follows
the Poisson distribution with mean $d$,
which is also the degree
distribution in $\mathcal{G}(n,\frac{d}{n})$.
For this reason, spectral algorithm based on
the adjacent
matrix $A$ or a constant power of $A$
or  modifications of such matrix
(say the Laplacian)
does not apply to sparse SBM. 

In sparse SBM,
we say the clustering problem is
\emph{solvable} iff there exists an estimator
of the community label, which is positively
 correlated to the true community label.
Most studies in
sparse SBM have been limited
to balance case.
\cite{decelle2011asymptotic} investigated
the sparse SBM with two communities
and balance
parameters i.e., $q = q^{(b)}=\begin{pmatrix}
\frac{\alpha}{n},\frac{\beta}{n}
\\ \frac{\beta}{n}, \frac{\alpha}{n}
\end{pmatrix},
\pi_0^{(b)}=\pi_1^{(b)}=1/2$. Based on
ideas from statistical physics
(cavity method), \cite{decelle2011asymptotic}
conjectured that
clustering  in $SBM(q^{(b)},\pi^{(b)})$
is solvable if and only if
$d\lambda^2>1$.
On the negative part,
\cite{mossel2015reconstruction}
showed that clustering
in $SBM(q^{(b)},\pi^{(b)})$ is not solvable
if $d\lambda^2<1$.
The same condition also implies
that the model $SBM(q^{(b)},\pi^{(b)})$
and the
Erd\H{o}s-R\'{e}ny model $\mathcal{G}(n,\frac{d}{n})$
are  contiguous
(which implies no consistent
estimator of $q^{(b)}$ exist).
On the positive part,
\cite{coja2010graph} provide
a spectral algorithm for clustering.
But in their paper,
the condition
ensuring the
positive correlation
 is stronger than the condition
$d\lambda^2>1$.
Finally,  \cite{mossel2013proof} and
\cite{massoulie2014community}
independently provide
spectral algorithms solving the
clustering problem
in $SBM(q^{(b)},\pi^{(b)})$ under the condition
$d\lambda^2>1$.
Therefore
\cite{mossel2015reconstruction},
\cite{mossel2013proof} and
\cite{massoulie2014community}
together confirmed the
conjecture proposed by \cite{decelle2011asymptotic}.
Recently, \cite{bordenave2015non} deal with the
general sparse SBM with arbitrarily many
blocks (see also \citep{abbe2015community}). Their
result confirm the "spectral redemption conjecture",
which is a generalized version of
the conjecture in \citep{decelle2011asymptotic}.
They prove, based on non-backtracking
walks on the graph, that community detection
is solvable down to the Kesten-Stigum threshold.

\cite{xu2014edge} studied
 the SBM with edge label.
The edge label indicates the type
of the connection.
For the SBM with edge label,
\cite{heimlicher2012community}
proposed a conjecture
similar to \citep{decelle2011asymptotic}.
 \cite{lelarge2015reconstruction},
similar to \cite{mossel2015reconstruction},
confirmed a half of the conjecture.
They
proved that the condition
proposed in \citep{heimlicher2012community}
implies that both consistent
parameter estimation and positively
correlated clustering are impossible.
On the positive part,
\cite{xu2014edge}
proposed a
clustering algorithm taking
advantage of the edge label. The proof of
positive correlation of their algorithm
only concerns
Chernoff inequality.
But the condition ensuring positive correlation
is stronger than that proposed by \citep{heimlicher2012community}.
 It is not known
whether  the spectral algorithms
in
\citep{mossel2013proof,massoulie2014community}
can be adapted
to provide a positively correlated
clustering algorithm under
the mere condition of \citep{heimlicher2012community}.
The   problem of estimating
the distribution of edge label is also
unknown.

\subsection{Motivation and technique}

  The technique used in dense SBM
can not be directly applied  to sparse SBM.
In dense SBM,  consistent parameter estimation
is usually a by product of consistent clustering.
But it does not seem that way in sparse SBM.
For instance, \cite{mossel2013proof}
uses the technique of random matrix
to estimate the community labels,
while
\cite{mossel2015reconstruction}
estimates $\lambda$ by counting $k-$cycles.
In dense SBM,
the lower bound of estimation
error is usually given by information
inequality such as Fano's inequality
 \citep{gao2015rate}.
In sparse SBM, second moment method,
which yields the results that two
models are closed, is used to prove impossibility
of parameter estimation
(see \citep{mossel2015reconstruction} section 5).
There is a good  reason to speculate that
 clustering is not solvable
 if the SBM is not distinguishable from some
  Erd\H{o}s-R\'{e}ny model
$\mathcal{G}(n,\frac{d}{n})$.
\citep{neeman2014non} recently obtained a result
in this fold.
Despite that \cite{bordenave2015non} has solved the
community detection problem and
the parameter estimation problem
for the general sparse SBM down to the
Kesten-Stigum threshold,
it is not known where exactly the threshold
for reconstructibility and distinguishability
is (see \citep{banks2016information} for such results).
 By far, most of results establishing
indistinguishability employ second moment
method. Hopefully, calculating the log partition
function of the Ising model on SBM provide
an alternative approach. Also note that
the conditional distribution of $\sigma$
given $G$ is approximately an Ising model
when $n$ is large. Therefore it is likely
that analysis of the Ising model
on a sparse SBM ultimately provide an exact
threshold for reconstructibility and distinguishability.

\subsection{Outline}
Denote the
  probability of the two
  communities by $\pi_0,\pi_1$.
  The
connection probability between
community $a$ and $b$ is $\frac{\alpha_{ab}}{n}$.
Since the graph is undirected,
it is required that $\alpha_{01}=
\alpha_{10}$.
Let
 $q = \begin{pmatrix}\frac{\alpha_{00}}{n} &
\frac{\alpha_{01}}{n}\\
\frac{\alpha_{10}}{n} &
\frac{\alpha_{11}}{n}\end{pmatrix}$,
let $d_0 = \pi_0\alpha_{00}+\pi_1\alpha_{01},
d_1=\pi_0\alpha_{10}+\pi_1\alpha_{11}$.
It is easy to see
if $d_0=d_1=d$, then there exists
$r\geq 0, \lambda$ with
$\pi_0=\frac{1}{1+r},\pi_1=\frac{r}{1+r}$,
and $q = \begin{pmatrix}
\frac{d(1+r\lambda)}{n}\ &
\frac{d(1-\lambda)}{n}\\
\frac{d(1-\lambda)}{n}\ &
\frac{d(1+\frac{\lambda}{r})}{n}\end{pmatrix}$.
We focus on the regime $d_0=d_1=d$. The
parameter estimation and community detection
in the regime
$d_0\ne d_1$ are trivial.
Let $SBM(d,\lambda,r)$ denote
the stochastic block model
defined by  $d,\lambda,r$.
 The paper is organized as follows.
 We show in section
 \ref{sbmsubsec1}
 that $d,\lambda$ can
be estimated in the same way
(by counting  $k-$cycles)
as
in the balanced SBM.
We evaluate
the log partition function of an Ising
model on graph $G$ in section \ref{sbmsubsec2}.
As an application,
we propose a consistent estimator of
 $r$ when $\lambda<0$ and $d$
being sufficiently large.
We provide 
a random clustering algorithm,
which samples $\hat{\sigma}$ according
to an appropriate Ising model
on $G$,
 in section \ref{sbmsubsec3}.
 The clustering algorithm has
  positive correlation with the true
community label when $\lambda<0$.
 Section \ref{sbmsec3}
contains proof of lemmas in section \ref{sbmsec2}.
Concluding remarks
and some further questions
are given in section \ref{sbmsec4}.


\subsection{Notations}

For a given undirected graph $G=(V,E)$ and
a node $u\in V$,
let $deg(u)$ denote the degree of
$u$ in $G$.
For $A,B\subseteq V(G)$,
$e_G(A,B) = \big\{
\ \{i,j\}\in E(G): i\in A, j\in B
\big\}$.
For an event $A$,
$I(A)$ denote the indicator function of
$A$.
For two sets $A,B$,
$A\Delta B$ denote
$A-B\ \cup\  B-A$.
For two sequences of reals
$f_n,g_n$ write $f_n\sim g_n$ if
$\lim\limits_{n\rightarrow\infty}
\frac{f_n}{g_n}=1$;
$f_n=\Omega(g_n)$ if
$f_n\geq k g_n$ for some positive real
$k$.
We write
$\oE_{X|Y}[f(X,Y)]$ or
$\oE_{X}[f(X,Y)]$
($\mbbP_{X|Y}( (X,Y)\in \mathcal{Z} )$
or $\mbbP_{X}((X,Y)\in\mathcal{Z})$) to denote
the expectation
(probability) with respect to $X$
 conditional on $Y$.
Write $\oE_{X\sim p}$
($\mbbP_{X\sim p}$) to denote
the expectation (probability)
when the distribution of $X$ is $p$.

\section{Main results}\label{sbmsec2}

Denote both $d_0,d_1$ by $d$.
Let $P =
\begin{pmatrix}
p_{00}& p_{01}\\
p_{10}& p_{11}
\end{pmatrix}= \begin{pmatrix}
\frac{\pi_0\alpha_{00}}{d_0}&
\frac{\pi_1\alpha_{01}}{d_0}\\
\frac{\pi_0\alpha_{10}}{d_1}&
\frac{\pi_1\alpha_{11}}{d_1}
\end{pmatrix}$.
Bear in mind that $P$ can be regard as a
markov transition matrix.
Let $\lambda=p_{11}+p_{00}-1$ denote the second large
eigenvalue of $P$.
In subsection \ref{sbmsubsec1} we show that, similar with
\citep{mossel2015reconstruction} section 3,
by counting   $k$-cycles for
appropriately large $k$
 we can estimate $\lambda$
consistently
provided $d\lambda^2>1$.
 We give in subsection \ref{sbmsubsec2}
a consistent estimator of
$r$ in the case $\lambda<0$; and
subsection \ref{sbmsubsec3} a random
 clustering
algorithm with positive correlation
with true labeling in the same case.

\subsection{Estimating $d,\lambda$}
\label{sbmsubsec1}
Let $C_k$ denote the number
of cycles of length $k$. The following
proposition says that $\lambda$ can be
consistently estimated by counting
$k$-cycles.
\begin{proposition}\label{sbmprop4}
Let $\h{d} =\frac{1}{n} \sum\limits_{u\in V}
deg(u)$.
Then $\h{d}$ is a $\sqrt{n}-$consistent
estimator of $d$.
 If $d\lambda^2>1$, $k=o(\log n)$,
then $\frac{(2kC_k -\h{d}^k)^{\frac{1}{k}}}{\hat{d}}$
is a consistent estimator of $\lambda$.
\end{proposition}
\begin{proof}
The $\sqrt{n}$-consistency of
$\h{d}$ is obvious.
Prove the second conclusion,
we compute
the probability that a given sequence
of different nodes $u_1,\cdots,u_{k}$ forms a cycle.
Set $u_{k+1}=u_1$.
Note that
\begin{align}\label{sbmeq1}
\oE\big[\ \prod\limits_{i=1}^{k}
X_{u_i u_{i+1}}\big]
=&\sum\limits_{\sigma\in \{0,1\}^{k+1},\sigma_{k+1} = \sigma_1}
\prod\limits_{i=1}^{k}
\pi_{\sigma_{i+1}}
\frac{\alpha_{\sigma_i\sigma_{i+1}}}{n}
\\ \nonumber
=&\frac{d^k}{n^k}\sum\limits_{\sigma\in \{0,1\}^{k+1},\sigma_{k+1} = \sigma_1}
\prod\limits_{i=1}^{k}
p_{\sigma_i\sigma_{i+1}}.
\end{align}
Think of $\sum\limits_{\sigma\in \{0,1\}^{k+1},
\sigma_{k+1}=\sigma_1=h}
\ \prod\limits_{i=1}^{k}
p_{\sigma_i\sigma_{i+1}}$
as the probability of the following event:
a markov chain with transition matrix
$P$ starting at $h$,  arrives at
$h$ after $k$ steps. 
Therefore continue
(\ref{sbmeq1}) we have
\begin{align}\label{sbmeq2}
\oE\big[\ \prod\limits_{i=1}^{k}
X_{u_i u_{i+1}}\big]
=&\frac{d^k}{n^k}\bigg(\ (1,0)P^k\begin{pmatrix}
1\\ 0\end{pmatrix}+
(0,1)P^k\begin{pmatrix}
0\\ 1\end{pmatrix}\
\bigg)
\\ \nonumber
=&\frac{d^k}{n^k}\cdot
Tr[P^k]
\\ \nonumber
=&\frac{d^k}{n^k}(1+\lambda^k).
\end{align}
Thus, for $k=o(\log n)$,
\begin{align}\nonumber
\oE\big[ C_k\big]=
\binom{n}k\cdot k!\frac{1}{2k}
\cdot\frac{d^k}{n^k}(1+\lambda^k)
=(1+O(\frac{1}{n})) \frac{d^k}{2k}(1+\lambda^k).
\end{align}
Similarly, for
$k=o(\log n)$,
$Var(C_k) = O(\frac{d^k}{k}(1+\lambda^k))$.
This is  given by \citep{mossel2015reconstruction}
 theorem 3.1
where $\oE\big[C_k(C_k-1)
\cdots (C_k-m)\big]$ is calculated.
Therefore, if $d\lambda^2>1$, $k=o(\log n)$,
then
$$
2kC_k -\h{d}^k=
(1+O(\frac{1}{n}))(d\lambda)^{k}
+O_p(\frac{d^k}{n}+\frac{k}{\sqrt{n}}+\sqrt{k}d^{k/2}
+\sqrt{k}(d\lambda)^{k/2})
=
(d\lambda)^k(1+ o_p(1)) .
$$
 Thus, if $d\lambda^2>1$, $k=o(\log n)$,
then $\frac{(2kC_k -\h{d}^k)^{\frac{1}{k}}}{\hat{d}}$
is a consistent estimator of $\lambda$.

\end{proof}

\subsection{Evaluating the log partition function}
\label{sbmsubsec2}
Let $SBM(d,\lambda,r,n)$ denote
the SBM
defined by $d,\lambda,r$ of
size $n$.
Clearly, the SBM defined by
$d,\lambda,r$ and $d,\lambda,\frac{1}{r}$
are identical. Therefore, without loss
of generality, assume $r\geq 1$. Also note
that  $1+r\lambda\geq 0$ is automatically
required since $\alpha_{00}\geq 0$.

For any undirected graph
$G$ and $\sigma\in \{0,1\}^{V(G)}$,
let $J(\sigma; G) =
\bigg| \bigg\{\ \{u,v\}:\{u,v\}\in E(G),
\sigma(u)=\sigma(v)
\bigg\}
 \bigg|$.
We evaluate the
following log partition function:
\begin{align}\nonumber
Z(\beta,G) = \log
\big(
\sum\limits_{\sigma\in \{0,1\}^{V(G)}}
e^{-\beta J(\sigma;G)}
\big ).
\end{align}

To state our main
results, we introduce the following
symbols.
Denote by
$g(z)$ the function
\begin{eqnarray}\label{sbmdefg}
g(z) = \left\{
\begin{aligned}
&\min\big\{\
z-(1-z)\log (1-z),\ (1+z)\log (1+z)-z
\ \big\}\ \ &\text{if }0<z<1;\\ \nonumber
&\infty \ &\text{if }z> 1.
\end{aligned}
\right.
\end{eqnarray}
Let
\begin{align}\nonumber
&C(r,\lambda) =
\inf\limits_{
0\leq x\leq 1,0\leq y\leq 1}
\bigg\{
[r\lambda(x-y)^2+
(x+ry-\frac{1+r}{2})^2]
+\frac{(1+r)^2}{4}
\bigg\};
\\ \nonumber
&y^* = \min\{\frac{r+1}{2(r+\lambda)},1\},\
x^*=0;
\
\\ \nonumber
&
\varepsilon_0= g^{-1}
(\frac{4\log 2\cdot (1+r)^2}{dC(r,\lambda)});
\\ \nonumber
&\epsilon_0 =
  \dfrac{2(r+1)^2\big(\frac{2\log 2}{\beta d}+
 2\varepsilon_0\big)}
 {\min\{r-2r\lambda-1, \frac{-\lambda (1+r)^2}{r+\lambda}\}}
 ;
\\ \nonumber
&\epsilon_1 =
\min\bigg\{
\frac{
(1+r)^2(\frac{2\log 2}{\beta d}+2\varepsilon_0)
}
{|2(r^2+r\lambda)y^*-r^2-r|}
,
\sqrt{\frac
{
2(1+r)^2
(\frac{2\log 2}{\beta d}+2\varepsilon_0)
}
{
r^2+r\lambda
}
}
\bigg\}.
\end{align}

\begin{condition}\label{sbmcond1}
\ \\
\begin{enumerate}
\item $0<\beta\leq 1$;


\item
\begin{enumerate}
\item
$(r+1)^2\big(\frac{2\log 2}{\beta d}+
 2\varepsilon_0\big)<
 \frac{(1+r)^2}{4}-C(r,\lambda)$;

\item
$\min\bigg\{
\frac{
(1+r)^2(\frac{2\log 2}{\beta d}+2\varepsilon_0)
}
{|2(r^2+r\lambda)y^*-r^2-r|}
,
\sqrt{\frac
{
2(1+r)^2
(\frac{2\log 2}{\beta d}+2\varepsilon_0)
}
{
r^2+r\lambda
}
}
\bigg\}
<\min\{\frac{r+1}{2(r+\lambda)},1\}
-\frac{r+1}{2r} ;$

\item $
(r+1)^2\big(\frac{2\log 2}{\beta d}+
 2\varepsilon_0\big)
 \leq
 \dfrac
 {(r^2+r\lambda)
\cdot
\big( \min\{
 r-2r\lambda-1, \frac{-\lambda (1+r)^2}{r+\lambda}
 \}
 \big)^2}
 {8r^2(1-\lambda)^2}
$;

\item $d> \frac{9}{2}\cdot
\frac{4\log 2\cdot (1+r)^2}{C(r,\lambda)}.$
\end{enumerate}

\end{enumerate}
\end{condition}

\begin{remark}
For any $r,\lambda$ there exists
sufficiently large $d$ and
sufficiently small $\beta$ satisfying
the condition \ref{sbmcond1}.
In the sense $d$ being
large and $r,\lambda$ being constant,
the $\beta$ we bear in mind satisfy the
follows:
$\beta=o(1) $, $\beta d\rightarrow\infty$,
$\varepsilon_0=O(\frac{1}{\sqrt{d}})=o(1)$,
$\epsilon_0 = O(\frac{1}{\sqrt{d}}
+\frac{1}{\beta d})=o(1)$,
$\epsilon_1=O(
\sqrt{\frac{1}{\beta d}
+\frac{1}{\sqrt{d}}}
)=o(1)$.
To get an intuition of these quantities,
the reader is referred to
theorem \ref{sbmth4},
lemma \ref{sbmlem2}.

\end{remark}
Let
$$
 C(d,r,\lambda,\beta)
=-\frac{\beta d}{(1+r)^2}\cdot
\bigg( \min\big\{\
\frac{(r-1)(1-\lambda)}{1+r}
,\frac{|\lambda|(1+r)^2}{4(r+\lambda)}\
\big\}
-\big(\ 2\beta+\frac{\log 2}{\beta d}
+ 12\max\{\epsilon_0,\epsilon_1\}\
\big)
\bigg).
$$
The following properties
of these quantities are needed.
\begin{proposition}\label{sbmprop2}
\ \\

If $\lambda<0$ then:
\begin{enumerate}
\item under condition \ref{sbmcond1} item 2,
$\epsilon_0\leq x\wedge
\frac{1}{2}\leq y$
or $|y-y^*|\geq \epsilon_1\wedge \frac{1}{2}
 \leq y$
implies
$$
\beta
\cdot \big\{\
[\lambda r(x-y)^2+(x+ry-\frac{1+r}{2})^2]
 \frac{d}{(1+r)^2}+\frac{d}{4}
 \big\}
 \geq \beta \frac{d}{(1+r)^2} C(r,\lambda)+ 2\log 2
 +2\beta\varepsilon_0 d.
 $$

 \item
\begin{align}\nonumber
&\min\big\{\
\frac{x^*}{1+r}d(1+r\lambda)+
\frac{y^*r}{1+r}d(1-\lambda),
\frac{1-x^*}{1+r}d(1+r\lambda)+
\frac{(1-y^*)r}{1+r}d(1-\lambda)
\ \big\}
=
\frac{1+r-y^*r+y^*r\lambda}{1+r}d,
\\ \nonumber
&\min\big\{\
\frac{x^*}{1+r}d(1-\lambda)+
\frac{y^*r}{1+r}d(1+\frac{\lambda}{r}),
\frac{1-x^*}{1+r}d(1-\lambda)+
\frac{(1-y^*)r}{1+r}d(1+\frac{\lambda}{r})
\ \big\}=
\frac{y^*(r+\lambda)}{1+r}d .
\end{align}
\item
\begin{equation}\nonumber
  C(r,\lambda) =
  \left\{
  \begin{aligned}
  &\frac{r^2+2r\lambda +1}{2}\hspace{1cm}
  &&\text{if }r\leq 1-2\lambda;\\ \nonumber
  &\frac{(r+2\lambda)(1+r)^2}{4(r+\lambda)}
  &&\text{if }
  r\geq 1-2\lambda.
  \end{aligned}
  \right.
  \hspace{1cm}\text{ And }
  C(r,\lambda)<\frac{(1+r)^2}{4}.
\end{equation}

\item
$x^*,y^*$ is a minimizer
of $
[r\lambda(x-y)^2+
(x+ry-\frac{1+r}{2})^2]
+\frac{(1+r)^2}{4}$, and
$y^* = 1$ iff $r\leq 1-2\lambda$.

\item
There exists
 two constants depending
 on $\lambda$, namely $C_1(\lambda),
 C_2(\lambda)>0$, such that for any
  $\lambda<0$,
 any $d\geq C_1(\lambda)$,
 any $\frac{2}{\sqrt{d}}\geq \beta\geq \frac{1}{2\sqrt{d}}$
 we have:
 \begin{enumerate}
 \item
  $d,r,\lambda,\beta$ satisfy
 condition \ref{sbmcond1}
 for all $r\geq 1$;
 \item
$r-1\geq \frac{C_2(\lambda)}{d^{\frac{1}{4}}}$
implies
$C(d,r,\lambda,\frac{1}{\sqrt{d}})<0$ and
$\forall 0\leq r'<r
\ (\int_{r'}^r C(d,t,\lambda,\frac{1}{\sqrt{d}})dt
<0)$.
\end{enumerate}
\end{enumerate}

\end{proposition}
The proof of proposition
\ref{sbmprop2} and  other
lemmas, propositions
in this subsection are all
 delayed to section \ref{sbmsec3}.
The following theorem establish the
upper derivative of the log partition
function with respect to $r$ when
$\lambda<0$.

\begin{theorem}\label{sbmth1}

Consider this function of
$d,\lambda,r$ and $\beta$,
$ \oE_{G\sim SBM(d,\lambda,r,n)}
 \big[\ \frac{1}{n}Z(\beta,G)\ \big]$.

 \begin{enumerate}

\item  Let $G_n\sim SBM(d,\lambda, r,n)$,
then $\big|
\frac{1}{n}Z(\beta,G_n)
-
\oE[\frac{1}{n}Z(\beta,G_n)]
\ \big|=O_p(\frac{1}{\sqrt{n}})$.

 \item If $\lambda<0$ then
 under condition \ref{sbmcond1}
 $$
\limsup\limits_{\delta\rightarrow 0}
\ \limsup\limits_{n\rightarrow\infty}
\frac{1}{\delta}
\bigg(\  \oE_{G\sim SBM(d,\lambda,r+\delta,n)}
 \big[\ \frac{1}{n}Z(\beta,G)\ \big]
 -\oE_{G\sim SBM(d,\lambda,r,n)}
 \big[\ \frac{1}{n}Z(\beta,G)\ \big]
 \bigg)
\leq
C(d,\lambda,r,\beta).
$$

\item
$\oE_{G\sim SBM(d,\lambda,r,n)}
 \big[\ \frac{1}{n}Z(\beta,G)\ \big]$
 is Lipschitz continuous
 in $d,\lambda,r$ and $\beta$ uniformly in $n$.
 i.e.,
\begin{align}\nonumber
&\limsup\limits_{\delta\rightarrow 0}
\ \limsup\limits_{n\rightarrow\infty}
\frac{1}{\delta}
\bigg|
\oE_{G\sim SBM(d+\delta,\lambda,r,n)}
 \big[\ \frac{1}{n}Z(\beta,G)\ \big]
 -
 \oE_{G\sim SBM(d,\lambda,r,n)}
 \big[\ \frac{1}{n}Z(\beta,G)\ \big]
 \bigg|,\\ \nonumber
&\limsup\limits_{\delta\rightarrow 0}
\ \limsup\limits_{n\rightarrow\infty}
\frac{1}{\delta}
\bigg|
\oE_{G\sim SBM(d,\lambda+\delta,r,n)}
 \big[\ \frac{1}{n}Z(\beta,G)\ \big]
 -
 \oE_{G\sim SBM(d,\lambda,r,n)}
 \big[\ \frac{1}{n}Z(\beta,G)\ \big]
 \bigg|,
 \\ \nonumber
& \limsup\limits_{\delta\rightarrow 0}
\ \limsup\limits_{n\rightarrow\infty}
\frac{1}{\delta}
\bigg|
\oE_{G\sim SBM(d,\lambda,r+\delta,n)}
 \big[\ \frac{1}{n}Z(\beta,G)\ \big]
 -
 \oE_{G\sim SBM(d,\lambda,r,n)}
 \big[\ \frac{1}{n}Z(\beta,G)\ \big]
 \bigg|,
 \\ \nonumber
& \limsup\limits_{\delta\rightarrow 0}
\ \limsup\limits_{n\rightarrow\infty}
\frac{1}{\delta}
\bigg|
\oE_{G\sim SBM(d,\lambda,r,n)}
 \big[\ \frac{1}{n}Z(\beta+\delta,G)\ \big]
 -
 \oE_{G\sim SBM(d,\lambda,r,n)}
 \big[\ \frac{1}{n}Z(\beta,G)\ \big]
 \bigg|\\ \nonumber
<&\infty.
\end{align}
\end{enumerate}
\end{theorem}
Before proving theorem \ref{sbmth1},
we give a direct application of theorem
\ref{sbmth1}
providing
the following consistent estimator of $r$.
\begin{corollary}\label{sbmcol1}
If $\lambda<0$,
 $d\geq C_1(\lambda), r-1\geq \frac{C_2(\lambda)}
{d^{\frac{1}{4}}}$, then
the following estimator of $r$ is
consistent
:
\begin{align}\nonumber
\hat{r} =
\bigg(
\oE_{G\sim SBM(\hat{d},\hat{\lambda},
r,n)}
 \big[\ \frac{1}{n}(Z(\frac{1}{\sqrt{\hat{d}}},G))\ \big]
\bigg)^{-1}
\big(
\frac{1}{n}
Z(\frac{1}{\sqrt{\hat{d}}}, G_n)
\big).
\end{align}
Here
$\oE_{G\sim SBM(d,\lambda,
r,n)}
 \big[\ \frac{1}{n}(Z(\beta,G))\ \big]$
 is regarded as a single variable
  function in $r$
 and $()^{-1}$ denote its
 inverse. Constants $C_1(\lambda),C_2(\lambda)$
 are defined in proposition \ref{sbmprop2}
 conclusion 5.

\end{corollary}
 \begin{proof}

By theorem \ref{sbmth1}
conclusion 2 and proposition \ref{sbmprop2}
conclusion 5, for all
 $r-1\geq \frac{C_2(
 \lambda)}{d^{\frac{1}{4}}},
 d\geq C_1(\lambda)$, for any
 $r'\geq 1$,
$ \oE_{G\sim SBM(d,\lambda,
r,n)}
 \big[\ \frac{1}{n}Z(\frac{1}{\sqrt{d}},G)\ \big]
 =\oE_{G\sim SBM(d,\lambda,
r',n)}
 \big[\ \frac{1}{n}Z(\frac{1}{\sqrt{d}},G)\ \big]$
 implies $r' = r$. So
  the inverse function
 is well defined at
 $\frac{1}{n}Z(\frac{1}{\sqrt{d}},
 G_n)$ with large probability.
 $\frac{1}{n}Z(\frac{1}{\sqrt{d}},G)$
 has fluctuation of order $\frac{1}{\sqrt{n}}$
 by theorem \ref{sbmth1} conclusion 1,
 so it is closed to its expectation
 with large probability.
 Finally the
 conclusion follows by noting
 that
 $\hat{d},\hat{\lambda}$ are
 consistent estimator of
 $d,\lambda$ and
  $\oE_{G\sim SBM(d,\lambda,
r,n)}
 \big[\ \frac{1}{n}Z(\frac{1}{\sqrt{d}},G)\ \big]$
 is continuous
in $d,\lambda$ by theorem \ref{sbmth1} conclusion 3.

 \end{proof}

 \begin{proof}[Proof of theorem \ref{sbmth1}]

Conclusion 1 of theorem \ref{sbmth1}
follows by  concentration inequality
such as Azuma's inequality and its proof is
therefore omitted.
Conclusion $3$ of theorem \ref{sbmth1}
follows in the same way
as conclusion $2$ and its proof is
 therefore omitted.
Now we focus on the proof of conclusion 2
of theorem \ref{sbmth1}.
 Through out the proof, fix $\delta$
 to be a sufficiently small
positive
constant which will be smaller
than any other constant whenever necessary.

 We prove theorem \ref{sbmth1} conclusion
 2 by evaluating
$$\oE_{G\sim SBM(d,\lambda,r+\delta,n)}
\big[
Z(\beta,G)
\big]-
\oE_{G\sim SBM(d,\lambda,r,n)}
\big[Z(\beta,G)
\big]\ $$
for small $\delta$. To this end,
we adopt the variational method.
We firstly construct a graph $\tilde{G}$.
Based on $\tilde{G}$, we then
inductively construct
three sequences of random graph
$G_{0,i},G_{1,i},
,G'_{1,j}$ by adding
nodes,
 deleting edges or
adding edges.
Through these sequences we obtain
two graphs $G_0',G_1'$
(see definition \ref{sbmdefG}).
We argue
 that
 the marginal law of
 $G_0'$ $(G_1')$ 
is  sufficiently close to
$SBM(d,\lambda, r, n)$
($
SBM(d,\lambda,r+\delta,n) $)
(see proposition \ref{sbmprop5}).
Finally we evaluate
$\oE\big[
Z(\beta,G_0')-
Z(\beta,G_1')
\big]$
by evaluating
$\oE\big[
Z(\beta,G_{h,i+1})-Z(\beta,G_{h,i})
\big]$ (lemma \ref{sbmmainlem1}),
$\oE\big[
Z(\beta,G'_{h,i+1})-Z(\beta,G'_{h,i})
\big]$ (lemma \ref{sbmlem3}),
for $h\in\{0,1\}$.

More specifically,
$\tilde{G}$ is generated according
to a SBM consisting of two
communities $N_0,N_1$ with
$|N_0| \approx \frac{n}{1+r}-\frac{\delta n}{(1+r)^2}$,
$|N_1|\approx \frac{nr}{1+r} $.
The within-community
connection probability
of $N_0,N_1$ are
$\frac{d(1+r\lambda)}{n}$,
$\frac{d(1+\frac{\lambda}{r})}{n}$
respectively;
and the across community
 connection probability
 is $\frac{d(1-\lambda)}{n}$.
$G'_0$ is constructed by adding
$\frac{\delta n}{(1+r)^2}$ many
nodes to community $N_0$ and connect
 each new node with every old
node according
to model $SBM(d,\lambda ,r)$. i.e.,
connect each new node with every old node,
say $u$,
with probability
$\frac{d(1+r\lambda)}{n}$ if $u\in N_0$,
probability
$\frac{d(1-\lambda)}{n}$ if
$u\in N_1$.
$G'_1$ is constructed by firstly
adding $\frac{\delta n}{(1+r)^2}$ many
nodes to community $N_1$ and
connect each new node with every old node
according to model $SBM(d,\lambda, r+\delta)$.
i.e., connect each new node with every
old node, namely $u$, with
probability $\frac{d(1-\lambda)}{n}$
if $u\in N_0$, probability
$\frac{d(1+\frac{\lambda}{r})}{n}$
if $u\in N_1$.
Then adjust the connection probability
of $G'_1$ by deleting $[\frac{d|\lambda| \delta n}{2(1+r)^2}]$
many edges
in $e_{G'_1}(N_0,N_0)$ uniformly at
random and adding
$[\frac{d|\lambda| \delta n}{2(1+r)^2}]$
many disconnected pairs, $\{u,v\}$, with
$u,v\in N_1$ uniformly at random. Note that
no new nodes are connected. The precise definition
is as follows.

\begin{definition}\label{sbmdefG}
[Construction of $\tilde{G},
G_{0,i},G_{1,i},G'_{1,i}$]


\begin{itemize}
\item $\tilde{G}$: $V(\tilde{G}) = N_0\cup N_1$
where
$N_0=\{1,2,\cdots, \frac{n}{1+r}
-[\frac{\delta n}{(1+r)^2}]\}$,
$
N_1=\{\frac{n}{1+r}
-[\frac{\delta n}{(1+r)^2}]+1,
\cdots, n
-[\frac{\delta n}{(1+r)^2}]\}$;

$\tilde{G}$ is the following
random graph:
presence of edges
are independent and  
\begin{equation}
\mbbP(\{u,v\}\in E(G)) =
\left\{
\begin{aligned}
&\frac{d(1+r\lambda)}{n} \hspace{1cm}
\text{if }u,v\in N_0,\\
&\frac{d(1+\frac{\lambda}{r})}{n}
\hspace{1cm}
\text{if }u,v\in N_1,\\
&\frac{d(1-\lambda)}{n} \hspace{1cm}
\text{if }u\in N_0,v\in N_1
\text{ or }u\in N_1,v\in N_0.\\
\end{aligned}
\right.
\end{equation}

\item $G_{0,0}=G_{1,0} = \tilde{G}$;
set $k_i=n-[\frac{\delta n}{(1+r)^2}]+i$.
For $i\leq [\frac{\delta n}{(1+r)^2}]$,

$G_{0,i+1}:$
$V(G_{0,i+1}) = V(G_{0,i})\cup
\{k_i\}$;
let $X_{0,iu},u\in V(\tilde{G})$
be independent random variables with
$X_{0,iu}\sim Bin(1,\frac{d(1+r\lambda)}{n})$
if $u\in N_0$ and
$X_{0,iu}\sim Bin(1,\frac{d(1-\lambda)}{n})$
if $u\in N_1$;
$E(G_{0,i+1})=E(G_{0,i})\cup
\big\{\ \{k_i,u\}: X_{0,iu}=1\ \big\}$.

$G_{1,i+1}:$
$V(G_{1,i+1}) = V(G_{1,i})\cup
\{k_i\}$;
let $X_{1,iu},u\in V(\tilde{G})$
be independent binary random variables with
$X_{1,iu}\sim Bin(1,\frac{d(1-\lambda)}{n})$
if $u\in N_0$ and
$X_{1,iu}\sim Bin(1,\frac{d(1+\frac{\lambda}{r})}{n})$
if $u\in N_1$;
$E(G_{1,i+1})=E(G_{1,i})\cup
\big\{\ \{k_i,u\}: X_{1,iu}=1\}$.

\item
$
G'_{1,0}=G_{1,[\frac{\delta n}{(1+r)^2}]}$.
For $i\leq
[\frac{d|\lambda| \delta n}{2(1+r)^2}]
$,

$G'_{1,i+1}:$
$V(G'_{1,i+1})= V(G'_{1,0})$;
$E(G'_{1,i+1})$ is obtained
by deleting
an edge in $e_{G'_{1,i}}
(N_0,N_0)$,
namely $\{u_i,v_i\}$,  uniformly at random.

Then for $[\frac{d|\lambda| \delta n}{2(1+r)^2}]
+1\leq j\leq 2
[\frac{d|\lambda| \delta n}{2(1+r)^2}]
$,
$G'_{1,j+1}$: $V(G'_{1,j+1}) = V(G'_{1,0})$;
$E(G'_{1,j+1})$ is obtained
by adding a disconnected pair $\{u_j,v_j\}$
with $u_j,v_j\in N_1$
to $E(G'_{1,j})$  uniformly at random.

\item Denote
$G_{0,[\frac{\delta n}{(1+r)^2}]}$
by $G'_0$
and $G'_{1,2[\frac{d|\lambda| \delta n}{2(1+r)^2}]}$
 by $G'_1$.
\end{itemize}
\end{definition}

\begin{proposition}\label{sbmprop5}
 \begin{align}\nonumber
& \bigg|
\oE[\ Z(\beta,G'_0)
\ ]-
\oE_{G\sim SBM(d,\lambda,r,n)}
[Z(\beta,G)]\ \bigg|,\\ \nonumber
&\bigg|
 \oE[\ Z(\beta,G'_{1})
\ ]-
\oE_{G\sim SBM(d,\lambda,r+\delta,n)}
[Z(\beta,G)]
\ \bigg|\\ \nonumber
=&O(\sqrt{n}+\delta^2 n).
\end{align}
\end{proposition}
The term $\sqrt{n}$ is due to
the fluctuation of community size.
The term $\delta^2 n$ is due
to the approximation error of connection
probability. For instance,
there is no connection among "new
nodes" in $G'_1,G'_0$ while the
expected number of edges among "new nodes"
 should be
$O(\delta^2 n)$.
For instance, the expected number of
edges in the two communities in $G'_1$
are $\frac{d(1+(r+\delta)\lambda)}
{2(1+r+\delta)^2}n+O(\delta^2n)$,
$\frac{d(1+\frac{\lambda}{r+\delta})(r+\delta)^2}
{2(1+r+\delta)^2}n+O(\delta^2 n)$.

Now we evaluate
$\oE\big[
Z(\beta,G_{h,i+1})-Z(\beta,G_{h,i})
\big],
\oE\big[
Z(\beta,G'_{1,i+1})-Z(\beta,G'_{1,i})
\big]$,
for $h\in\{0,1\}$.
Denote by
IS($\beta,G$)  the following Ising
model
 on $\{0,1\}^{V(G)}$:
$$\mbbP(\sigma)= \frac{e^{-\beta J(\sigma;G)}}
{\sum\limits_{\gamma\in
\{0,1\}^{V(G)}}e^{-\beta J(\gamma;G)} }.$$
For $h,l\in\{0,1\}$,
$\sigma\in\{0,1\}^{V(G_{h,i})}$, let
$e_{il}^h(\sigma) =\big|e_{G_{h,i+1}}(k_i,\sigma^{-1}(l))
\big|$.
Recall from definition \ref{sbmdefG}
that $k_i$ is the node added into
$G_{h,i}$ at step $i$.
The key observation is:
\begin{align}\label{sbmeq7}
& Z(\beta,G_{h,i+1})-Z(\beta,G_{h,i})\
\\ \nonumber
=&
 \log\bigg(
\sum\limits_{\sigma\in \{0,1\}
^{V(G_{h,i})}}\dfrac
{
e^{-\beta J(\sigma; G_{h,i})}
\cdot
(e^{-\beta
e^h_{i0}(\sigma)
}+
e^{
-\beta e^h_{i1}(\sigma)
}
)
}
{
\sum\limits
_{\gamma\in \{0,1\}^{V(G_{h,i})}
}
e^{-\beta J(\gamma; G_{h,i})}
}
\bigg)\
\\ \nonumber
=& \log\big(
\
\oE_{\sigma\sim IS(\beta,G_{h,i})}
[
\ e^{-\beta\cdot e^h_{i0}}+
e^{-\beta \cdot e^h_{i1}}\
]
\ \big)
.
\end{align}

Here and below,
 $e^h_{il}$ is short for $e^h_{il}(\sigma)$.
Another key point is to take
advantage of the convexity of $\log$
as follows:
\begin{align}\label{sbmeqlog}
\oE_{X,Y}[\log f(X,Y)]
\leq \oE_X\big[\ \log \big(\
\oE_{Y|X}[ f(X,Y)]\ \big)\
\big]
\leq \log \oE_{X,Y}[ f(X,Y)].
\end{align}

Therefore using (\ref{sbmeq7})
(\ref{sbmeqlog}) we have:
\begin{align}\label{sbmeq8}
&\oE_{G_{h,i+1}|G_{h,i}}
\big[\ Z(\beta,G_{h,i+1})-Z(\beta,G_{h,i})\ \big]
\\ \nonumber
=&
\oE_{G_{h,i+1}|G_{h,i}}
\bigg[
\log\bigg(
\oE_{\sigma|G_{h,i}}\big[
\ e^{-\beta\cdot e^h_{i0}}+
e^{-\beta \cdot e^h_{i1}}\
\big]
\bigg)\
\bigg]
\\ \nonumber
\leq& \log \oE_{G_{h,i+1},\sigma|G_{h,i}}
\big[
\ e^{-\beta\cdot e^h_{i0}}+
e^{-\beta \cdot e^h_{i1}}\
\big].
\end{align}

\begin{align}\label{sbmeq10}
&\oE_{G_{h,i+1}|G_{h,i}}
\big[\ Z(\beta,G_{h,i+1})-Z(\beta,G_{h,i})\ \big]
\\ \nonumber
\geq&
\oE_{G_{h,i+1},\sigma|G_{h,i}}
\big[
\ \log
\big(
\ e^{-\beta\cdot e^h_{i0}}+
e^{-\beta \cdot e^h_{i1}}\
\big)
\ \big]
\\ \nonumber
\geq&
\oE_{G_{h,i+1},\sigma|G_{h,i}}
[\
-\beta \min\big\{\  e^h_{i0},
 e^h_{i1}\
\big\}
\ ]
.
\end{align}
Where $\sigma|G_{h,i}$ is the Ising model
$IS(\beta,G_{h,i})$.
Note that  in the calculation
of $\oE_{G_{h,i+1},\sigma|G_{h,i}}$,
$\sigma,G_{h,i+1}$ conditioned on $G_{h,i}$
are mutually independent.
It is not surprise that we need some properties
on the Ising model on $G_{h,i}$.
Let
\begin{align}\label{sbmdefxy}
&\overline{l} =
\operatorname*{arg\, max}\limits_{l\in\{0,1\}}
\frac{|\sigma^{-1}(l)\cap N_1|}{|N_1|};
\\ \nonumber
&x(\sigma) =
\frac{|\sigma^{-1}(\overline{l})
\cap N_{0}|}{|N_{0}|},\ \
y(\sigma) = \frac{|\sigma^{-1}(\overline{l})
\cap N_{1}|}{|N_{1}|}.
\end{align}

We prove that for $G_{h,i},G'_{1,j}$,
with large probability (with respect to
$G_{h,i},G'_{1,j}$):
$y(\sigma)$ ( $x(\sigma)$ ) is
 closed to
$y^*$ ($ x^*$)
with large probability
(with respect to $\sigma$).
By condition \ref{sbmcond1}
item 2-(d), $\varepsilon_0$ is well defined
and therefore $\epsilon_0,\epsilon_1$
are well defined. So the following
lemmas make sense.
\begin{lemma}\label{sbmlem2}
Assume condition \ref{sbmcond1} holds.
If $\lambda<0$, then
for any $i\leq [\frac{\delta n}{(1+r)^2}]$,
$j\leq 2[\frac{d|\lambda| \delta n}{2(1+r)^2}]$,
$h\in\{0,1\}$,
with probability larger than $1-4\cdot2^{-n}$:
\begin{align}\nonumber
&\mbbP_{\sigma\sim IS(\beta,G_{h,i})}
\bigg(\ \epsilon_0
< x(\sigma)
 \bigg),\ \
 \mbbP_{\sigma\sim IS(\beta,G'_{1,j})}
\bigg(\ \epsilon_0
< x(\sigma)
 \bigg)\leq 2^{-n},
\\ \nonumber
&\mbbP_{\sigma\sim IS(\beta,G_{h,i})}
\bigg(\
\big|
y(\sigma)-y^*
\big|>\epsilon_1
 \ \bigg),\ \
  \mbbP_{\sigma\sim IS(\beta,G'_{1,j})}
\bigg(\ \big|
y(\sigma)-y^*
\big|>\epsilon_1
 \bigg)
 \leq 2^{-n}.
 \end{align}

\end{lemma}

Combine (\ref{sbmeq8})(\ref{sbmeq10})
with lemma \ref{sbmlem2} and after some
tedious calculation, we are able to
evaluate
$\oE\big[
Z(\beta,G_{h,i+1})-Z(\beta,G_{h,i})
\big],
\oE\big[
Z(\beta,G'_{h,i+1})-Z(\beta,G'_{h,i})
\big]$.
\begin{lemma}\label{sbmmainlem1}
\begin{align}\label{sbmmaineq1}
&\limsup\limits_{n\rightarrow\infty}
\sup\limits_{i\leq
[\frac{\delta n}{(1+r)^2}]
}
\oE
\big[\ Z(\beta,G_{0,i+1})-Z(\beta,G_{0,i})\ \big]
\leq
\log 2+(e^{-\beta}-1)d\cdot
[\ -2\max\{\epsilon_0,\epsilon_1\}+
\frac{1+r-y^*r+y^*r\lambda}{1+r}
\ ],
\\ \nonumber
&\limsup\limits_{n\rightarrow\infty}
\sup\limits_{i\leq
[\frac{\delta n}{(1+r)^2}]
}
\oE
\big[\ Z(\beta,G_{1,i+1})-Z(\beta,G_{1,i})\ \big]
\leq
\log 2+
(e^{-\beta}-1)d\cdot
[\ -2\max\{\epsilon_0,\epsilon_1\}+
\frac{y^*(r+\lambda)}{1+r}
\ ].
\end{align}
\begin{align}\label{sbmmaineq2}
&\liminf\limits_{n\rightarrow\infty}
\sup\limits_{i\leq
[\frac{\delta n}{(1+r)^2}]
}
\oE
\big[\ Z(\beta,G_{0,i+1})-Z(\beta,G_{0,i})\ \big]
\geq
-\beta d\cdot (2\max\{\epsilon_0,\epsilon_1\}+
\frac{1+r-y^*r+y^*r\lambda}{1+r}),
\\ \nonumber
&\liminf\limits_{n\rightarrow\infty}
\sup\limits_{i\leq
[\frac{\delta n}{(1+r)^2}]
}
\oE
\big[\ Z(\beta,G_{1,i+1})-Z(\beta,G_{1,i})\ \big]
\geq
-\beta d\cdot (2\max\{\epsilon_0,\epsilon_1\}+
\frac{y^*(r+\lambda)}{1+r}d
).
\end{align}
\end{lemma}

\begin{lemma}\label{sbmlem3}
\begin{align}\label{sbmeqmain3}
&\limsup\limits_{n\rightarrow\infty}
\ \sup\limits_{
[\frac{d|\lambda|\delta n }{2(1+r)^2}]+1\leq
j\leq
2[\frac{d|\lambda|\delta n }{2(1+r)^2}]
}\
\oE\big[\ Z(\beta,G'_{1,j+1})-
Z(\beta,G'_{1,j})\ \big]
\\ \nonumber
\leq&
(e^{-\beta}-1)
\big(\
1- 2y^*(1-y^*)
-
2\epsilon_1\ \big)
.\end{align}
\end{lemma}

Now we can prove theorem \ref{sbmth1}
conclusion 2.
By lemma \ref{sbmmainlem1},
\begin{align}\label{sbmeq20}
&\oE\big[
Z(\beta,G'_{0})
-Z(\beta,\tilde{G})
\big]
\geq
-\beta d\cdot
\big(
2\max\{\epsilon_0,\epsilon_1\}+
\frac{1+r-y^*r+y^*r\lambda}{1+r}\
\big)\cdot
\frac{\delta n}{(1+r)^2},
\\ \nonumber
&\oE\big[
Z(\beta,G'_{1,0})
-Z(\beta,\tilde{G})
\big]
\leq
\big[ \log 2+
(e^{-\beta}-1)d\cdot
(\ -2\max\{\epsilon_0,\epsilon_1\}+
\frac{y^*(r+\lambda)}{1+r})\
\big]\cdot
\frac{\delta n}{(1+r)^2}.
\end{align}
It is  obvious that,
\begin{align}\label{sbmeq21}
\oE\big[
Z(\beta,G'_{1,[\frac{d|\lambda|\delta n }{2(1+r)^2}]})
-Z(\beta,G'_{1,0})
\big]
\leq \beta d|\lambda|
\frac{\delta n}{2(1+r)^2}.
\end{align}
And by lemma \ref{sbmlem3},
\begin{align}\label{sbmeq22}
\oE\big[
Z(\beta,G'_1)
-
Z(\beta,G'_{1,[\frac{d|\lambda|\delta n }{2(1+r)^2}]})
\big]\leq
(e^{-\beta}-1)
\big(\
1- 2y^*(1-y^*)
-
2\epsilon_1\ \big)
\cdot d|\lambda|
\frac{\delta n}{2(1+r)^2}.
\end{align}
In summary of
(\ref{sbmeq20})(\ref{sbmeq21})(\ref{sbmeq22}),
\begin{align}\label{sbmeq23}
&\oE\big[
Z(\beta,G'_0)
-
Z(\beta,G'_1)
\big]\\ \nonumber
\geq&
-\beta d\cdot
\big(
2\max\{\epsilon_0,\epsilon_1\}+
\frac{1+r-y^*r+y^*r\lambda}{1+r}\
\big)\cdot
\frac{\delta n}{(1+r)^2}
\\ \nonumber
&\hspace{1cm}
-
\big[ \log 2+
(e^{-\beta}-1)d\cdot
\big(\ -2\max\{\epsilon_0,\epsilon_1\}+
\frac{y^*(r+\lambda)}{1+r}
\big)\
\big]\cdot
\frac{\delta n}{(1+r)^2}\\ \nonumber
&\hspace{1cm}-
\beta d|\lambda|
\frac{\delta n}{2(1+r)^2}
\\ \nonumber
&\hspace{1cm}-
(e^{-\beta}-1)
\big(\
1- 2y^*(1-y^*)
-
2\epsilon_1\ \big)
\cdot d|\lambda|
\frac{\delta n}{2(1+r)^2}
\\ \nonumber
\geq &
\frac{\delta d n}{(1+r)^2}
\cdot
\bigg(
-\beta \cdot
\big(\
\frac{1+r-y^*r+y^*r\lambda}{1+r}
-\frac{y^*(r+\lambda)}{1+r}
+|\lambda|y^*(1-y^*)
\ \big)
\\ \nonumber
&\hspace{2cm}
-2(e^{-\beta}+\beta-1)
-\frac{\log 2}{d}
-4(1-e^{-\beta}+\beta)\cdot
\max\{\epsilon_0,\epsilon_1\}\
\bigg).
\end{align}

Intuitively,
the dominating term is
$-\beta \cdot
\big(\
\frac{1+r-y^*r+y^*r\lambda}{1+r}
-\frac{y^*(r+\lambda)}{1+r}
+2|\lambda|y^*(1-y^*)
\ \big)$, which
is of order $\beta$.
It is helpful to recall that
$\beta=o(1),
\max\{\epsilon_0,\epsilon_1\} =
o(1),\frac{1}{d} = o(1),
\beta d>>1 $.
The last three terms
are of order
$O(\beta^2 ), O(1/d),
o(\beta)$, and are thus
ignorable compared to $\beta$.

The dominating  term is,
\begin{align}\nonumber
&r\leq 1-2\lambda\ \ \Rightarrow
\ \
y^*=1\ \
\Rightarrow \hspace{1cm}
&&\ \ \frac{1+r-y^*r+y^*r\lambda}{1+r}
-\frac{y^*(r+\lambda)}{1+r}
+|\lambda|y^*(1-y^*)
\\ \nonumber
&\
&&=\frac{(r-1)(\lambda-1)}{1+r}<0;
\\ \nonumber
&r>1-2\lambda\ \ \Rightarrow\ \
y^*=\frac{r+1}{2(r+\lambda)}
\Rightarrow\hspace{1cm}
&&\ \ \frac{1+r-y^*r+y^*r\lambda}{1+r}
-\frac{y^*(r+\lambda)}{1+r}
+|\lambda|y^*(1-y^*)
\\ \nonumber
&\
&&=\frac{\lambda(1+r)^2}{4(r+\lambda)^2}<0
.\end{align}
Also note
by condition \ref{sbmcond1} item 1,
 $e^{-\beta}+\beta-1\leq \beta^2$
 and
$\beta^2\leq\beta$
.
 Thus continue (\ref{sbmeq23}),
\begin{align}\nonumber
&\oE\big[
Z(\beta,G'_0)
-
Z(\beta,G'_1)
\big]
\\ \nonumber
\geq&
\delta n\cdot \frac{\beta d}{(1+r)^2}\cdot
\bigg(\ \min\big\{\
\frac{(r-1)(1-\lambda)}{1+r}
,\frac{|\lambda|(1+r)^2}{4(r+\lambda)}\
\big\}
-\big(\ 2\beta+\frac{\log 2}{\beta d}
+ 12\max\{\epsilon_0,\epsilon_1\}\
\big)\ \
\bigg).
\end{align}
The conclusion 2 of theorem \ref{sbmth1}
thus follows.

\end{proof}

\subsection{Clustering when $\lambda<0$}
\label{sbmsubsec3}
Recall that $\pi_0=\frac{1}{1+r}$,
$\pi_1=\frac{r}{1+r}$ and $r\geq 1$.
We provide the following
random clustering algorithm:
given the observed graph $G_n\sim SBM(d,\lambda,r,n)$,
sample a $\sigma\sim
IS(\frac{1}{\sqrt{\hat{d}}},G_n)$;
let $l'= \operatorname*{arg\, max}
\limits_{l\in\{0,1\}}
|\sigma^{-1}(l)|$;
the estimator for the community label
is, $\tau(u)= I(u\in \sigma^{-1}(l')) $.
Let $M_0,M_1$ denote the two communities
of $G_n$ with $\oE[|M_0|] = \pi_0n$.
Recall from proposition \ref{sbmprop2}
conclusion 5 the
definition of $C_1(\lambda)$.
\begin{theorem}\label{sbmth4}
If $\lambda<0$,
$d\geq C_1(\lambda)$, then
the estimator $\tau(\cdot)$
is positively correlated to the
true labeling since
\begin{align}\nonumber
\mbbP\bigg(
\
\ \big|\
\frac{|\sigma^{-1}(l')\cap M_0|}
{|M_0|}
- \frac{x^*}{1+r}
\big|\leq \epsilon_0,
 \big|\
\frac{|\sigma^{-1}(l')\cap M_1|}
{|M_1|}
- \frac{y^* r}{1+r}
 \big|\leq \epsilon_1
\
\bigg)
=1-e^{-\Omega(n)}.
\end{align}

\end{theorem}
\begin{proof}
By proposition \ref{sbmprop2}
conclusion 5
and consistency of $\hat{d}$,
 condition \ref{sbmcond1}
holds for
$\beta=\frac{1}{\sqrt{\hat{d}}},d,r,\lambda$
with large probability.
Therefore,
by lemma \ref{sbmlem2}
\begin{align}\nonumber
\mbbP\bigg(
\ (\exists l\in\{0,1\})
\ \big|\
\frac{|\sigma^{-1}(l)\cap M_0|}
{|M_0|}
- \frac{x^*}{1+r}
\big|\leq \epsilon_0,
 \big|\
\frac{|\sigma^{-1}(l)\cap M_1|}
{|M_1|}
- \frac{y^* r}{1+r}
 \big|\leq \epsilon_1
\
\bigg)
\geq 1-2^{-n+1}.
\end{align}
But for the $l\in\{0,1\}$ with
$\big|\
\frac{|\sigma^{-1}(l)\cap M_0|}
{|M_0|}
- x^*
\big|\leq \epsilon_0,
 \big|\
\frac{|\sigma^{-1}(l)\cap M_1|}
{|M_1|}
- y^*
 \big|\leq \epsilon_1$,
 by large deviation theorem for
$|M_1|$,
 we have that with probability
 $1-e^{-\Omega(n)}$,
 $\frac{|\sigma^{-1}(l)|}{n}
 >\frac{1}{2}$ since by condition
 \ref{sbmcond1} item 2-(b)
 $\frac{(y^*-\epsilon_1) r}{1+r}
>\frac{1}{2}$. i.e.,
with probability $1-e^{-\Omega(n)}$,
$l=l'$.
The proof is thus accomplished.

\end{proof}

In practice, a variety  of techniques
 are available to sample $\sigma\sim
IS(\beta,G)$, for example the MCMC sampling.

\section{Proof of lemmas}
\label{sbmsec3}

\subsection{Proof of proposition \ref{sbmprop2}}

Conclusions $2,3,4$ follow by
direct computation.
We only give a sketched proof of item $1$ and $5$.
Let $f(x,y) = [\lambda r(x-y)^2+(x+ry-\frac{1+r}{2})^2]
+\frac{(1+r)^2}{4}$. It suffices
to show that the condition
for $x,y$ implies
$f(x,y)-C(r,\lambda)\geq
(r+1)^2\big(\frac{2\log 2}{\beta d}+
 2\varepsilon_0\big)$.

 For the first half of conclusion $1$,
suppose the minimum of
$f(x,y)$ on $x\in [\epsilon_0,1],y\in[\frac{1}{2},1]$
is attained at $x',y'$,
then either $x'$ lies
on the boundary of $[\epsilon_0,1]$
or $y'$ lies on the boundary of
$[\frac{1}{2},1]$.
Therefore it is easy to check that the minimum
must be
attained at either
$(\epsilon_0, \tilde{y})$ or
$ (\frac{1}{2},\frac{1}{2})$
for some $\tilde{y}$.
Clearly, we only  have
to deal with the case $(\epsilon_0, \tilde{y})$
since $f(\frac{1}{2},\frac{1}{2})
\geq C(r,\lambda)+
(r+1)^2\big(\frac{2\log 2}{\beta d}+
 2\varepsilon_0\big)$
  by condition \ref{sbmcond1}
 item 2-(a).
If $\tilde{y}=y^*$ then the conclusion
follows by definition of $\epsilon_0$
($=2\dfrac{(r+1)^2\big(\frac{2\log 2}{\beta d}+
 2\varepsilon_0\big)}
{\frac{\partial}{\partial x}f(0,y^*)}$)
and the fact
$\frac{\partial }{\partial x}f(0,y^*)>0$
,$\frac{\partial^2}{\partial x^2}
f(x,y)=2(1+r\lambda)>0$.
If $\tilde{y} = \frac{1}{2}$ then
$f(\epsilon_0,\tilde{y})\geq f(\frac{1}{2},
\frac{1}{2})\geq
C(r,\lambda)+(r+1)^2\big(\frac{2\log 2}{\beta d}+
 2\varepsilon_0\big)$.
If $\tilde{y}\ne y^*\wedge \tilde{y}>\frac{1}{2}$
then
it must be the case $\frac{\partial}{\partial y}
f(\epsilon_0,\tilde{y})=0\wedge
\frac{\partial}{\partial y}
f(\epsilon_0,y^*)>0$
since $\frac{\partial}{\partial y}
f(0,y^*)\leq 0$ and
$\frac{\partial^2 f(x,y)}{
\partial x\partial y}>0$.
This implies
 $\tilde{y}\leq y^*$ since
 $\frac{\partial^2}{\partial y^2}f(x,y)>0$.
Moreover,
$$\frac{\partial}{\partial y}
f(\epsilon_0,\tilde{y})
-(\tilde{y}-y^*)\frac{\partial^2}{\partial y^2}
f(x,y)-
\epsilon_0\frac{\partial^2}{\partial y\partial x}
f(x,y)
=\frac{\partial}{\partial y}
f(0,y^*)\leq 0.$$
So
$y^*-\tilde{y}\leq
\epsilon_0\frac{r(1-\lambda)}{r^2+r\lambda}
$. Also note that
for any $x\in [0, \epsilon_0],
y\in[\tilde{y},y^*]$,
$\frac{\partial }{\partial y}
f(x,y)\leq \epsilon_0\cdot
\frac{\partial^2}{\partial x\partial y}
f(x,y)=\epsilon_0\cdot 2r(1-\lambda)$.
Therefore,
$f(\epsilon_0,y^*)-
f(\epsilon_0,\tilde{y})
\leq (y^*-\tilde{y})
\frac{\partial}{\partial y}
f(\epsilon_0,y^*)\leq
\epsilon_0^2\frac{2r^2(1-\lambda)^2}{r^2+r\lambda}
\leq
(r+1)^2\big(\frac{2\log 2}{\beta d}+
 2\varepsilon_0\big)
$
(the last inequality
follows from
 condition \ref{sbmcond1} item 2-(c)).
Therefore by definition of $\epsilon_0$,
$$
f(\epsilon_0,\tilde{y})-C(r,\lambda)
\geq \epsilon_0\frac{\partial}{\partial x}
f(0,y^*)-
(f(\epsilon_0,y^*)- f(\epsilon_0,\tilde{y}))
\geq
(r+1)^2\big(\frac{2\log 2}{\beta d}+
 2\varepsilon_0\big).
$$

 For the second half of the conclusion
 $1$, note that
 by condition \ref{sbmcond1}
 item 2-(b)
 (which implies
 $2r(1-\lambda)(y^*-\epsilon_1)-(r+1)>0$)
  and the fact
 $\frac{\partial^2}{\partial x\partial y}
 f(x,y)=2r(1-\lambda)>0$, we have
 $
 0<\frac{\partial}{\partial x}
 f(0,y^*-\epsilon_1)<
 \frac{\partial}{\partial x}
 f(0,y^*+\epsilon_1)
 $.
  So by the fact
  $\frac{\partial^2 f(x,y)}{\partial^2 y},
  \frac{\partial^2 f(x,y)}{\partial x\partial y}>0
  $ the minimum
of $f(x,y)$ on
$x\in[0,1],|y-y^*|\geq \epsilon_1,y\in[\frac{1}{2},1]$
is attained at either of following points:
$(0,y^*-\epsilon_1)$,
 $(\frac{1}{2},\frac{1}{2})$,
 and $
(0,y^*+\epsilon_1)$
(if $y^*+\epsilon_1\leq 1$ of course).
 It is clear that by definition of $\epsilon_1$,
 $$\frac{1}{2}\epsilon_1^2\cdot
 \frac{\partial^2}{\partial y^2}
 f(x,y), \epsilon_1\big|
 \frac{\partial}{\partial y}
 f(0,y^*)\big|> (r+1)^2\big(\frac{2\log 2}{\beta d}+
 2\varepsilon_0\big).$$
 Thus,
 $f(0,y^*-\epsilon_1)\geq
 C(r,\lambda)+ (r+1)^2\big(\frac{2\log 2}{\beta d}+
 2\varepsilon_0\big)$
 and $f(0,y^*+\epsilon_1)
\geq C(r,\lambda)+ (r+1)^2\big(\frac{2\log 2}{\beta d}+
 2\varepsilon_0\big)$ if
 $y^*+\epsilon_1\leq 1$.
 By condition \ref{sbmcond1} item 2-(a),
 $f(\frac{1}{2},\frac{1}{2})
 \geq C(r,\lambda)+(r+1)^2\big(\frac{2\log 2}{\beta d}+
 2\varepsilon_0\big)$.
 Thus the second half
of  conclusion 1 follows.

Now we prove conclusion 5.
Note that, for any $\lambda<0$,
the following quantities
 are bounded away from
$0$ on $r\in [1,\infty)$:
$\min\{\frac{r+1}{2(r+\lambda)},1\}
-\frac{r+1}{2r},
\frac{(1+r)^2}{4}-C(r,\lambda),
r^2+r\lambda $,
$\min\{
 r-2r\lambda-1, \frac{-\lambda (1+r)^2}{r+\lambda}
 \}$,
 $C(r,\lambda)$.
 Therefore,
 if
 $\frac{1}{2\sqrt{d}}\leq \beta \leq \frac{2}{\sqrt{d}}$,
 then
  there exists
 a constant depending
 on $\lambda$, namely $C_1(\lambda)>0$, such that
 for any $r\geq 1$,
 any $d\geq C_1(\lambda)$,
 $d,r,\lambda$ satisfy
 condition \ref{sbmcond1}.

  If  $\frac{1}{2\sqrt{d}}\leq \beta \leq \frac{2}{\sqrt{d}}$,
 then  for large $d$, we have,
 uniformly in $r$,
 $\epsilon_0=O(\frac{1}{\sqrt{d}})$,
$ \varepsilon_0=O(\frac{1}{\sqrt{d}})$,
 $\epsilon_1= O(\frac{1}{d^{\frac{1}{4}}})$.
  Note that when $r$ is close to $1$,
 it is possible that
  $C(d,r,\lambda,\frac{1}{\sqrt{d}})>0$
  since $\min\big\{\
\frac{(r-1)(1-\lambda)}{1+r}
,\frac{|\lambda|(1+r)^2}{4(r+\lambda)}\
\big\} = O(r-1)$.
But obviously,
there exists
 a constant $C_2(\lambda)$
 such that for any $d\geq C_1(\lambda)$,
 any $r-1\geq \frac{C_2(
 \lambda)}{d^{\frac{1}{4}}}$ :
 \begin{enumerate}
 \item
$C(d,\lambda,r,\frac{1}{\sqrt{d}})<0$;
 \item
 $
 \forall 0\leq r'<r\ (\int_{r'}^r
 C(d,\lambda,t,\frac{1}{\sqrt{d}})dt
 <0\ ).
 $
 \end{enumerate}
Thus the conclusion follows.

\subsection{Proof of proposition \ref{sbmprop5}}

Let $X_{ij}
(Y_{ij}),1\leq i<j\leq n$ denote
a set of random variables
whose joint distribution is
the law of
$I(\{i,j\}\in E(G'_0)),1\leq i<j\leq n$
($I(\{i,j\}\in E(G'_1)),1\leq i<j\leq n$
).

To prove the conclusion for
$G'_{0}$,
let
$G\sim SBM(d,\lambda,r,n)$;
let
$Z_{ij} = I(\{i,j\}\in E(G))$
and denote by $M_0,M_1$ the two random
communities of $G$. Without loss
of generality suppose
$$
M_1 = \{\frac{n}{1+r}-[\frac{\delta n}{(1+r)^2}]+1,
\frac{n}{1+r}-[\frac{\delta n}{(1+r)^2}]+2,
\cdots,\frac{n}{1+r}-[\frac{\delta n}{(1+r)^2}]+|M_1|\}.
$$
Note that $X_{ij},i,j\leq n$
are mutually independent.
It is easy
to see that
we can couple
$X_{ij},1\leq i<j\leq  n$
with $Z_{ij},1\leq i<j\leq  n$
in the following way,
\begin{itemize}
\item $(X_{ij},Z_{ij}),i,j\leq n$
are mutually independent;
\item if $\oE[X_{ij}]\leq \oE[Z_{ij}]$,
then $\mbbP(Z_{ij}=1|X_{ij}=1)=1$,
$\mbbP(Z_{ij}=1|X_{ij}=0)=
\frac{\oE[Z_{ij}-X_{ij}]}{1-\oE[X_{ij}]}$;

\item if $\oE[X_{ij}]> \oE[Z_{ij}]$,
then $\mbbP(X_{ij}=1|Z_{ij}=1)=1$,
$\mbbP(Z_{ij}=1|X_{ij}=0)=
\frac{\oE[X_{ij}-Z_{ij}]}{1-\oE[Z_{ij}]}$.
\end{itemize}

Let $N_\Delta = \{n-[\frac{\delta n}{(1+r)^2}]+1,
\cdots,n\}$.
Then we have,
$$\oE\big[
\ \sum\limits_{1\leq i<j\leq n}
\big|X_{ij}-Z_{ij} \big|
\ \big]
\leq
2
d\cdot\bigg|\ |M_1|- |N_1|\ \bigg|
+O(\delta^2 n).
$$

Thus,
\begin{align}\nonumber
&\oE_{M_0,M_1}
\bigg[
\ \oE\bigg[\ \big|Z(\beta,G'_0)
-
Z(\beta,G)
\big|\ \bigg| M_0,M_1 \bigg]\
\bigg]\\ \nonumber
\leq&
|\beta|\cdot\oE_{M_0,M_1}
\bigg[
 \ \oE\big[
\ \sum\limits_{1\leq i<j\leq n}
\big|X_{ij}-Z_{ij} \big|
\ \big] \
\bigg]
\\ \nonumber
\leq&O(\delta^2 n)+
2d\cdot\oE_{M_0,M_1}
\bigg[\
\bigg|\ |M_1|- |N_1|\ \bigg|
\ \bigg]
\\ \nonumber
=&O(\delta^2 n+\sqrt{n}).
\end{align}
The conclusion for $G'_0$ is done.

The proof for $G'_1$ is a little
harder because
$I(\{i,j\}\in E(G'_1))$
are not mutually independent.
Let $G\sim SBM(d,\lambda,r+\lambda,n)$
and $M_0 ,
M_1$ be the two communities
as above.
The coupling of
$Z_{ij},Y_{ij}$ when
$(i,j)\in (N_0\times N_1)\cup
(N_1\times N_0)$
or
 $i\in N_\Delta$ or $j\in N_\Delta$
are the same as previous.
In order to couple the
rest of $Z_{ij},Y_{ij}$,
consider an auxiliary graph
$G''_1$ which is obtained
from $G_{1,[\frac{\delta n}{(1+r)^2}]}$
in the following way:
delete each edge
in $e_{G'_{1,[\frac{\delta n}{(1+r)^2}]}}
(N_0,N_0)$  independently
with probability
$\frac{\delta|\lambda|}{1+r\lambda}$;
add each disconnected pairs
in $N_1$ independently
to $E(G'_{1,[\frac{\delta n}{(1+r)^2}]})$
with probability
$\frac{\delta d|\lambda|}{r^2 n}$.
Let $Y''_{ij} = I(\{i,j\}\in E(G''_1))$.
Note that $Y''_{ij}$ are
mutually independent.
So clearly we can couple
$Y''_{ij},
1\leq i<j\leq n$ with $Z_{ij}
, 1\leq i<j\leq n$ such that
$$
\oE\bigg[
\sum\limits_{1\leq i<j\leq n}
\big|
Y''_{ij}-
Z_{ij}
\big|\ \bigg|
M_0,M_1
\bigg]
\leq
2d\cdot\bigg|\ |M_1|- |N_1|-|N_\Delta|\ \bigg|
+O(\delta^2 n).
$$
Therefore
$$
\oE\big[\
\big|Z(\beta,G_1'')-
Z(\beta,G)\big|\
\big]=O(\sqrt{n}+\delta^2 n).
$$
Now it remains to
couple $Y''_{ij},
Y_{ij}$ for
$(i,j)\in (N_0\times N_0)
\cup (N_1\times N_1)$
. We demonstrate the coupling
for $Y''_{ij},Y_{ij}$
with $(i,j)\in N_0\times N_0$
as follows.
\begin{itemize}
\item To generate $Y_{ij}''$, let $Y''\sim Bin(
\big|e_{G'_{1,0}}(N_0,N_0)\big|,
\frac{\delta d|\lambda|}{1+r\lambda})$;
\item delete
 $Y''$ many edges
in $e_{G'_{1,0}}(N_0,N_0)$
uniformly at random;
\item
$Y''_{ij}=1$ if and only if $\{i,j\}
\in e_{G'_{1,0}}(N_0,N_0)$
and $\{i,j\}$ is not deleted.

\item To generate $Y_{ij}$,
let $E_{delete}$ denote
the set of $Y''$ deleted edges.

 If $Y''\geq
 \min\big\{\
 [\frac{\delta n d|\lambda|}{2(1+r)^2}]
 ,\big|e_{G'_{1,0}}(N_0,N_0)
 \big|\
 \big\} $,
 select
 $Y''-\min\big\{\
 [\frac{\delta n d|\lambda|}{2(1+r)^2}]
 ,\big|e_{G'_{1,0}}(N_0,N_0)
 \big|\
 \big\}$ many edges in
 $E_{delete}$ uniformly at random.
 For $i,j\in N_0$, set $Y_{ij} = 1$
 if and only if, $\{i,j\}\in
 e_{G'_{1,0}}(N_0,N_0)-E_{delete}$ ,
 or $\{i,j\}\in
 E_{delete}$ but is selected;

  If $Y''<
 \min\big\{\
 [\frac{\delta n d|\lambda|}{2(1+r)^2}]
 ,\big|e_{G'_{1,0}}(N_0,N_0)
 \big|\
 \big\} $,
  select
 other than $E_{delete}$,
 a set of
 $\min\big\{\
 [\frac{\delta n d|\lambda|}{2(1+r)^2}]
 ,\big|e_{G'_{1,0}}(N_0,N_0)
 \big|\
 \big\} -Y''$ edges
 from
 $e_{G'_{1,0}}(N_0,N_0)$
 uniformly at random.
  For $i,j\in N_0$, set $Y_{ij} = 1$ if and only if
  $\{i,j\}\in e_{G'_{1,0}}(N_0,N_0)$
  and is not selected.
\end{itemize}
Clearly
\begin{align}\nonumber
\oE\bigg[
\sum\limits_{\{i,j\}
\in e_{G'_{1,0}}(N_0,N_0)}
\big|
Y_{ij}-
Y''_{ij}
\big|\
\bigg|\ G_{1,0}',Y''
\bigg]
\leq
\bigg|
\min\big\{\
 [\frac{\delta n d|\lambda|}{2(1+r)^2}]
 ,\big|e_{G'_{1,0}}(N_0,N_0)
 \big|\
 \big\} -Y''
 \bigg|.
\end{align}
However, it is obvious
that
\begin{align}\nonumber
\oE\bigg[
\ \bigg|\min\big\{\
 [\frac{\delta n d|\lambda|}{2(1+r)^2}]
 ,\big|e_{G'_{1,0}}(N_0,N_0)
 \big|\
 \big\} -Y''
 \bigg|\
\bigg]
=O(\sqrt{n}+\delta^2 n).
\end{align}
Thus the proof is accomplished.

\subsection{Proof of lemma \ref{sbmlem2}}

We only prove that
with probability larger than $1-2^{-n}$,
 $\mbbP_{\sigma\sim IS(\beta,G_{h,i})}
\bigg(\ \epsilon_0
< x(\sigma)
 \bigg)<2^{-n}$. The other conclusions follow
 in the same fashion.

Firstly we prove a large deviation
result for $J(\sigma;G)$ with
$G$ generated by a SBM. Recall
from (\ref{sbmdefg}) the definition
of $g(\cdot)$.
Fix any $\sigma\in \{0,1\}^{V(G)}$,
we have for any $0\leq \varepsilon$,
\begin{align}\label{sbmeq3}
&\mbbP_G\bigg(
\ \big|\ J(\sigma;G)-\oE_G[J(\sigma;G)]
\ \big|
> \varepsilon \oE_G[J(\sigma;G)]
\ \bigg)\\ \nonumber
\leq& \exp\bigg
\{
-\oE_G[J(\sigma;G)]\cdot
g(\varepsilon)
 \bigg\}.
\end{align}
The proof of (\ref{sbmeq3}) follows by
standard use of Chernoff inequality
and a calculation of
$\oE_G(e^{\theta J(\sigma;G)})$ as
follows.
Let $\mbbP(\{u,v\}\in E(G)) = \frac{p_{uv}}{n}$,
for any $\theta\in\mathbb{R}$,
\begin{align}\nonumber
\oE_G[e^{\theta J(\sigma;G)}]
 =& \prod\limits_{u<v\in V(G)}
 ((1+(e^\theta-1)\frac{p_{uv}}{n})
 \\ \nonumber
 \leq&
  \exp\bigg\{
 (e^\theta-1)\sum\limits_{u<v\in V(G)}
 \frac{p_{uv}}{n}
 \bigg\}
 \\ \nonumber
 =&\exp
 \bigg\{(e^\theta-1)\oE_G[J(\sigma;G)]
\bigg\}.
\end{align}

Combine (\ref{sbmeq3}) with
Borel Cantali's lemma we have,
\begin{align}\label{sbmeq5}
&\mbbP
\bigg(
\
(\exists \sigma\in \{0,1\}^{V(G_{h,i})})
\ \big|J(\sigma; G_{h,i})-\oE_{G_{h,i}}(J(\sigma;G_{h,i}))\big|
> \varepsilon
\cdot\big(\ \sup\limits_{\sigma}
\oE_{G_{h,i}}[J(\sigma;G_{h,i})]
\ \big)
 \ \bigg)\\ \nonumber
 \leq& 2^n
  \exp\bigg\{\
  -g(\varepsilon)\cdot
\big( \inf\limits_{\sigma}
 \oE_{G_{h,i}}[J(\sigma;G_{h,i})]
 \ \big)
\  \bigg\}.
\end{align}

We take advantage of the following
evaluation for
$\frac{1}{n}\oE_{G_{h,i}}[J(\sigma;G_{h,i})]
$ (which clearly follows from
the construction of $G_{h,i}$):
\begin{align}\label{sbmeq4}
\bigg|\
\frac{1}{n}\oE_{G_{h,i}}[J(\sigma;G_{h,i})]&-
\frac{1}{2}\begin{pmatrix}
\frac{x(\sigma)}{1+r},
\frac{ry(\sigma)}{1+r}
\end{pmatrix}
\begin{pmatrix}
&d(1+r\lambda) &d(1-\lambda)\\
&d(1-\lambda) &d(1+\frac{\lambda}{r})
\end{pmatrix}
\begin{pmatrix}
\frac{x(\sigma)}{1+r}\\
\frac{ry(\sigma)}{1+r}
\end{pmatrix}
\\ \nonumber
&- \frac{1}{2}\begin{pmatrix}\frac{(1-x(\sigma))}{1+r},
\frac{r(1-y(\sigma))}{1+r}\end{pmatrix}
\begin{pmatrix}
&d(1+r\lambda) &d(1-\lambda)\\
&d(1-\lambda) &d(1+\frac{\lambda}{r})
\end{pmatrix}
\begin{pmatrix}
\frac{1-x(\sigma)}{1+r}\\
\frac{r(1-y(\sigma))}{1+r}
\end{pmatrix}\ \bigg|\
\\ \nonumber
=&
\bigg|\ \frac{1}{n}\oE_{G_{h,i}}[J(\sigma;G_{h,i})]-
[r\lambda(x(\sigma)-y(\sigma))^2+
(x(\sigma)+ry(\sigma)-\frac{1+r}{2})^2]\frac{d}{(1+r)^2}
-\frac{d}{4}
\bigg|
\\ \nonumber
\leq& \delta d.
\end{align}

Thus using (\ref{sbmeq4}) (and
since $\delta$ is sufficiently small)
\begin{align}\label{sbmeq14}
&\inf\limits_{\sigma \in
 \{0,1\}^{V(G_{h,i})}}
 \frac{1}{n}\oE_{G_{h,i}}[
 J(\sigma;G_{h,i})]
 \geq
 \frac{d}{(1+r)^2} C(r,\lambda)-d\delta
 \geq d\frac{C(r,\lambda)}{2(1+r)^2}
\\ \nonumber
&\sup\limits_{\sigma \in
 \{0,1\}^{V(G_{h,i})}}
 \frac{1}{n}\oE_{G_{h,i}}[J(\sigma;G_{h,i})]
 \leq
  \frac{d}{2}+d\delta\leq \frac{3d}{4}.
  \end{align}
  Substituting (\ref{sbmeq14}) into
  (\ref{sbmeq5}),
  \begin{align}\label{sbm5-2}
&\mbbP
\bigg(
\
(\exists \sigma\in \{0,1\}^{V(G_{h,i})})
\ \big|\
J(\sigma; G_{h,i})-\oE_{G_{h,i}}[J(\sigma;G_{h,i})]
\ \big|
> \varepsilon \frac{3}{4}d n
 \bigg)\\ \nonumber
 \leq&
 \exp\bigg\{\  \log 2\cdot n
 -g(\varepsilon)\cdot
d \frac{C(r,\lambda)}{2(1+r)^2}\cdot n \ \bigg\}
  .\end{align}
By condition \ref{sbmcond1}
item 2-(d)
$\varepsilon_0$ is well defined.
So substituting
$\varepsilon$ by $\varepsilon_0$
in (\ref{sbm5-2}) and
 by definition of $\varepsilon_0$ we have,
\begin{align}\label{sbmeq6}
\mbbP
\bigg(
\
(\exists \sigma\in \{0,1\}^{V(G_{h,i})})
\ \big|\ J(\sigma; G_{h,i})-\oE_{G_{h,i}}(J(\sigma;G_{h,i}))
\ \big|
> \frac{3\varepsilon_0}{4} dn
\ \bigg)
\leq e^{-\log 2\cdot n}
.\end{align}

Now we can prove lemma \ref{sbmlem2}.
Clearly for any $\sigma'$,
\begin{align}\label{sbmeq34}
&\mbbP_{\sigma\sim IS(\beta,G_{h,i})}
\bigg(\
\epsilon_0
< x(\sigma)
\bigg)
<
 \dfrac{
 \sum
 \limits_{
 \epsilon_0
< x(\sigma)
}
 \exp\{-\beta J(\sigma; G_{h,i})\}}
 {\exp\{-\beta J(\sigma';G_{h,i})\}}.
\end{align}
Set
$\sigma'$ to be any element of $\{0,1\}^{V(G_{h,i})}$
satisfying:
\begin{align}\nonumber
\bigg|\ \oE_{G_{h,i}}[
J(\sigma';G_{h,i})
]-
\inf\limits_{0\leq x,y\leq 1}
\big\{\
[\lambda r(x-y)^2+(x+ry-\frac{1+r}{2})^2]
 \frac{d}{(1+r)^2}+\frac{d}{4}\
\big\}
\bigg|
\leq \delta d.
\end{align}
By (\ref{sbmeq4}) such $\sigma'$ exists.
By (\ref{sbmeq6}),
with probability larger
than $1-2^{-n}$,
for all $\sigma$,
approximating $J(
\sigma;G_{h,i})$ by
$\oE_{G_{h,i}}[ J(\sigma;G_{h,i}) ]$
introduce an error smaller than
$\frac{3\varepsilon_0 d}{4}$.
Moreover, by (\ref{sbmeq4}),
for all $\sigma$
approximating
$\oE_{G_{h,i}}[
 J(\sigma;G_{h,i}) ]$ by
$[r\lambda(x(\sigma)-y(\sigma))^2+
(x(\sigma)+ry(\sigma)-\frac{1+r}{2})^2]\frac{d}{(1+r)^2}
+\frac{d}{4}$
introduce an error smaller than $\delta d$. 
Therefore,
using (\ref{sbmeq34}) we have,
for all $i\leq [\frac{\delta n}{(1+r)^2}]$,
the following event
occurs with probability
larger than $1-2^{-n}$:
\begin{align}\nonumber
&\mbbP_{\sigma\sim IS(\beta,G_{h,i})}
\bigg(\
\epsilon_0
< x(\sigma)
\
 \bigg)
 \\ \nonumber
 \leq & 2^n\exp\bigg\{-\beta n
 \bigg(\ \inf\limits_{
 \epsilon_0 <x\leq 1, 1/2\leq y\leq 1}
 \big\{\ [\lambda r(x-y)^2+(x+ry-\frac{1+r}{2})^2]
 \frac{d}{(1+r)^2}+\frac{d}{4}\big\} - d\delta -
 d\frac{3\varepsilon_0}{4}
 \ \bigg)
 \bigg\}\\ \nonumber
 &\cdot\exp\bigg\{\ \beta n \bigg(
\inf\limits_{
0\leq x,y\leq 1}
 \big\{\ [\lambda r(x-y)^2+(x+ry-\frac{1+r}{2})^2]
 \frac{d}{(1+r)^2}+\frac{d}{4}\ \big\}
 +d\delta +d\frac{3\varepsilon_0}{4} \bigg)
 \ \bigg\}
 \\ \nonumber
 &\text{by proposition \ref{sbmprop2}
 conclusion  1
 and
 since } 2\delta <\frac{\varepsilon_0}{2}
 \\ \nonumber
 \leq &\exp
 \bigg\{
 \ \big[\log 2 + 2\beta\varepsilon_0 d
 -\big(\ 2\beta\varepsilon_0 d
 +2\log 2\ \big)
 \big]\cdot n
 \ \bigg\}
 \\ \nonumber
 =& 2^{-n}.
\end{align}
Similarly, with probability larger
than $1-2^{-n}$
\begin{align}\nonumber
&\mbbP_{\sigma\sim IS(\beta,G_{h,i})}
\bigg(\ \big|
y(\sigma)-y^*
\big|>\epsilon_1\
 \bigg)
 \leq 2^{-n}.
 \end{align}
Thus the proof is accomplished.

\subsection{Proof of lemma \ref{sbmmainlem1}}

We demonstrate the
proof of (\ref{sbmmaineq1}) by
   analyzing (\ref{sbmeq8}) for $h=0$.
Fix an arbitrary
$n,i$, we have to evaluate
$
\oE_{G_{0,i+1},\sigma|G_{0,i}}
\big[
\ e^{-\beta\cdot e^0_{i0}}
+e^{-\beta \cdot e^0_{i1}}
\ \big]$.
Note that $e^h_{il} =
\sum\limits_{j\in \sigma^{-1}(l)}
X_{h,ij}$.
But $X_{h,ij}$ are mutually independent whose
 distribution does not concern $\sigma$.
Therefore,
\begin{align}\label{sbmeq9}
&\oE_{G_{0,i+1}|
G_{0,i},\sigma}
\big[\ e^{-\beta e^0_{i0}}
\ \big]
\\ \nonumber
\leq & \exp
\bigg\{
(e^{-\beta}-1)\cdot
(\ \frac{|\sigma^{-1}(0)\cap N_0|}{n}
d(1+r\lambda)+
\frac{|\sigma^{-1}(0)\cap N_1|}{n}
d(1-\lambda)
\ )\
\bigg\},\\ \nonumber
&
\oE_{G_{0,i+1}|
G_{0,i},\sigma}
\big[
\ e^{-\beta e^0_{i1}}]
\\ \nonumber
\leq &
\exp
\bigg\{(e^{-\beta}-1)\cdot
(\ \frac{|\sigma^{-1}(1)\cap N_0|}{n}
d(1+r\lambda)+
\frac{|\sigma^{-1}(1)\cap N_1|}{n}
d(1-\lambda)
\ )
\bigg\}.
\end{align}

Note that,
\begin{align}\label{sbmeq24}
\big|\frac{|\sigma^{-1}(h)\cap N_0|}{n}
- \frac{|\sigma^{-1}(h)\cap N_0|/|N_0|}{1+r}
\big|,
\ \
\big|\frac{|\sigma^{-1}(h)\cap N_1|}{n}
- \frac{r|\sigma^{-1}(h)\cap N_1|/|N_1|}{1+r}
\big|
\leq \frac{r\delta}{(1+r)^2}
.\end{align}

By lemma \ref{sbmlem2},
with probability larger than $1-2^{-n}$,
$\mbbP_{\sigma\sim IS(\beta, G_{h,i})}
\big(\
|x(\sigma)-x^*|\leq \epsilon_0,
|y(\sigma)-y^*|\leq \epsilon_1
\big)\geq 1-2^{-n}$.
Therefore continue
(\ref{sbmeq9}) and
approximate $\frac{|\sigma^{-1}(\overline{l})
\cap N_0|}{n}$
with $\frac{x(\sigma)}{1+r}$,
$\frac{|\sigma^{-1}(\overline{l})
\cap N_1|}{n}$
with
$\frac{y(\sigma)r}{1+r}$
(recall from (\ref{sbmdefxy}) the definition of $\overline{l}$,
$x(\sigma),y(\sigma)$).
By (\ref{sbmeq24}),
we have that
with probability $1-e^{-\Omega(n)}$:
for all $i\leq [\frac{\delta n}{(1+r)^2}]$,
\begin{align}\label{sbmeq11}
&\oE_{G_{0,i+1},\sigma|G_{0,i}}
\big[
\ e^{-\beta\cdot e^0_{i0}}
+e^{-\beta \cdot e^0_{i1}}
\
\big]
\\ \nonumber
\leq&
\oE_{G_{0,i+1},\sigma|G_{0,i}}
\big[
\ e^{-\beta\cdot e^0_{i\ 1-\overline{l}}}
+e^{-\beta \cdot e^0_{i\overline{l}}}
\ \big|\ \
|x(\sigma)-x^*|\leq \epsilon_0,
|y(\sigma)-y^*|\leq \epsilon_1
\big]\\ \nonumber
&+
2e^{2d(1-\lambda)}\cdot
\mbbP_{\sigma\sim IS(\beta,G_{0,i})}
\big(\
|x(\sigma)-x^*|>\epsilon_0\vee
|y(\sigma)-y^*|>\epsilon_1
\big)
\\ \nonumber
\leq&
\oE_{\sigma|G_{0,i}}
\bigg[\
\exp
\bigg\{
(e^{-\beta}-1)
[\ -d\delta+
\frac{1-x(\sigma)}{1+r}
d(1+r\lambda)
+
\frac{
(1-y(\sigma))r
}{1+r}
dr(1-\lambda)
\ ]
\bigg\}\\ \nonumber
&\hspace{1.2cm}
+\exp
\bigg\{
(e^{-\beta}-1)
[\ -d\delta+
\frac{x(\sigma)}{1+r}
d(1+r\lambda)
+
\frac{
y(\sigma)r
}{1+r}
dr(1-\lambda)
\ ]
\bigg\} \ \ \bigg|
\ |x(\sigma)-x^*|\leq \epsilon_0,
|y(\sigma)-y^*|\leq \epsilon_1
\ \bigg]+
 e^{-\Omega(n)}
\\ \nonumber
\leq&
\exp
\bigg\{(e^{-\beta}-1)
(-\max\{\epsilon_0,\epsilon_1\}d-\delta d)
\bigg\}
\\ \nonumber
&\hspace{0.3cm}
\cdot\bigg[\
\exp\bigg\{
(e^{-\beta}-1)
[\frac{x^*}{1+r}d(1+r\lambda)+
\frac{y^*r}{1+r}d(1-\lambda)]
\bigg\}
\\ \nonumber
&\hspace{0.8cm}+\exp\bigg\{
(e^{-\beta}-1)
[\frac{1-x^*}{1+r}d(1+r\lambda)+
\frac{(1-y^*)r}{1+r}d(1-\lambda)
]
\bigg\}\
\bigg]
+e^{-\Omega(n)}
\\ \nonumber
&\text{by proposition \ref{sbmprop2}
conclusion 2}
\\ \nonumber
\leq&
2\exp\bigg\{
(e^{-\beta}-1)
\big[\ -2\max\{\epsilon_0,\epsilon_1\}+
\frac{1+r-y^*r+y^*r\lambda}{1+r}
\ \big]\cdot d\ \bigg\}.
\end{align}
Similarly,
with probability $1-e^{-\Omega(n)}$:
for all $i\leq [\frac{\delta n}{(1+r)^2}]$,
\begin{align}\label{sbmeq12}
&\oE_{G_{1,i+1},\sigma|G_{1,i}}
\big[
\ e^{-\beta\cdot e^1_{i0}}
+e^{-\beta \cdot e^1_{i1}}
\ \big]
\\ \nonumber
\leq&
\exp
\bigg\{\ (e^{-\beta}-1)
(-\max\{\epsilon_0,\epsilon_1\}d-\delta d)
\bigg\}
\\ \nonumber
&\hspace{0.3cm}
\cdot\bigg[\
\exp\bigg\{
(e^{-\beta}-1)
[\frac{x^*}{1+r}d(1-\lambda)+
\frac{y^*r}{1+r}d(1+\frac{\lambda}{r})]
\bigg\}
\\ \nonumber
&\hspace{0.8cm}+\exp\bigg\{
(e^{-\beta}-1)
[\frac{1-x^*}{1+r}d(1-\lambda)+
\frac{(1-y^*)r}{1+r}d(1+\frac{\lambda}{r})
]
\bigg\}\
\bigg]
+e^{-\Omega(n)}
\\ \nonumber
&\text{ by proposition \ref{sbmprop2} conclusion 2}
\\ \nonumber
\leq&
2\exp
\bigg\{\
(e^{-\beta}-1)
\big[\ -2\max\{\epsilon_0,\epsilon_1\}
+\frac{y^*(r+\lambda)}{1+r}
\ \big]\cdot d
\bigg\}.
\end{align}

Substitute (\ref{sbmeq12})(\ref{sbmeq11})
into (\ref{sbmeq8}) we have,
with probability $1-e^{-\Omega(n)}$:
for all $i\leq [\frac{\delta n}{(1+r)^2}]$,
\begin{align}\nonumber
&\oE_{G_{0,i+1}|G_{0,i}}
\big[\ Z(\beta,G_{0,i+1})-Z(\beta,G_{0,i})\ \big]
\\ \nonumber
\leq&
\log 2+(e^{-\beta}-1)
[\ -2\max\{\epsilon_0,\epsilon_1\}+
\frac{1+r-y^*r+y^*r\lambda}{1+r}
\ ]\cdot d,
\\ \nonumber
&\oE_{G_{1,i+1}|G_{1,i}}
\big[\ Z(\beta,G_{1,i+1})-Z(\beta,G_{1,i})\ \big]
\\ \nonumber
\leq&
\log 2+
(e^{-\beta}-1)
[\ -2\max\{\epsilon_0,\epsilon_1\}+
\frac{y^*(r+\lambda)}{1+r}
\ ]\cdot d.
\end{align}
Thus the conclusion of (\ref{sbmmaineq1})
follows.

Proving (\ref{sbmmaineq2}) is similar.
Using lemma \ref{sbmlem2},
proposition \ref{sbmprop2} conclusion 2 and
 approximating
$\frac{|\sigma^{-1}(\overline{l})
\cap N_0|}{n}$
with $\frac{x(\sigma)}{1+r}$,
$\frac{|\sigma^{-1}(\overline{l})
\cap N_1|}{n}$
with
$\frac{y(\sigma)r}{1+r}$,
in the same way as
 (\ref{sbmeq11}),
we have that with probability $1-e^{-\Omega(n)}$:
for all $i\leq [\frac{\delta n}{(1+r)^2}]$,
\begin{align}\label{sbmeq13}
\oE_{G_{0,i+1},\sigma|G_{0,i}}
\big[\ \min\{ e^0_{i0},e^0_{i1}\}\ \big]
=& \oE_{G_{0,i+1},\sigma|G_{0,i}}
\big[\ \min\{ e^0_{i\ 1-\overline{l}},e^0_{i\overline{l}}\}\ \big]
\\ \nonumber
\leq&
\min\bigg\{
\oE_{G_{0,i+1},\sigma|G_{0,i}}
\big[\ e^0_{i\ 1-\overline{l}}\ \big],
\oE_{G_{0,i+1},\sigma|G_{0,i}}
\big[\ e^0_{i \overline{l}}\ \big]
\bigg\}
\\ \nonumber
\leq&
\big[\ 2\max\{\epsilon_0,\epsilon_1\}
+
\frac{1+r-y^*r+y^*r\lambda}{1+r}
\ \big]\cdot d,
\\ \nonumber
\oE_{G_{1,i+1},\sigma|G_{1,i}}
\big[\ \min\{ e^1_{i0},e^1_{i1}\}\ \big]
=& \oE_{G_{1,i+1},\sigma|G_{1,i}}
\big[\ \min\{ e^1_{i\ 1-\overline{l}},e^1_{i\overline{l}}\}\ \big]
\\ \nonumber
\leq &\min
\bigg\{
\oE_{G_{1,i+1},\sigma|G_{1,i}}
\big[
\  e^1_{i\ 1-\overline{l}}
\ \big]
,
\oE_{G_{1,i+1},\sigma|G_{1,i}}
\big[
\ e^1_{i\overline{l}}
\ \big]
\bigg\}
\\ \nonumber
\leq&
\big[\
2\max\{\epsilon_0,\epsilon_1\}
+\frac{y^*(r+\lambda)}{1+r}
\big] \cdot
d.
\end{align}
Substituting (\ref{sbmeq13})
into (\ref{sbmeq10}), we have that
with probability $1-e^{-\Omega(n)}$:
for all $i\leq [\frac{\delta n}{(1+r)^2}]$,
\begin{align}\nonumber
&\oE_{G_{0,i+1}|G_{0,i}}
\big[\ Z(\beta,G_{0,i+1})-Z(\beta,G_{0,i})\ \big]
\\ \nonumber
\geq&
-\beta\ \big[
2\max\{\epsilon_0,\epsilon_1\}+
\frac{1+r-y^*r+y^*r\lambda}{1+r}
\big] d,
\\ \nonumber
&\oE_{G_{1,i+1}|G_{1,i}}
\big[\ Z(\beta,G_{1,i+1})-Z(\beta,G_{1,i})\ \big]
\\ \nonumber
\geq&
-\beta\ \big[
2\max\{\epsilon_0,\epsilon_1\}+
\frac{y^*(r+\lambda)}{1+r}
\big]
d.
\end{align}
Thus the conclusion of (\ref{sbmmaineq2})
follows.

\subsection{Proof of lemma \ref{sbmlem3}}

Recall that
$\{u_j,v_j\}$ denote the
edge added at step $j$
in the construction of $G_1'$.
Note that,
\begin{align}\nonumber
&Z(\beta,G'_{1,j+1})-
Z(\beta,G'_{1,j})
=
\log\bigg(
\oE_{\sigma\sim
IS(\beta, G'_{1,j})}
[e^{-\beta I(\sigma(u_j)=\sigma(v_j))}]
\bigg).
\end{align}

Using convexity of $\log$
(\ref{sbmeqlog})
as  in (\ref{sbmeq8}),
we have,
\begin{align}\label{sbmeq36}
\oE_{G'_{1,j+1}|G'_{1,j}}
\big[\ Z(\beta,G'_{1,j+1})-
Z(\beta,G'_{1,j})\ \big]
\leq
\log\bigg(
\oE_{G'_{1,j+1},\sigma|
G'_{1,j}}
\big[
e^{-\beta I(\sigma(u_j)=\sigma(v_j))}
\big]
\bigg).
\end{align}
And in the calculation of
$\oE_{G'_{1,j+1},\sigma|
G'_{1,j}}$, $\sigma,G'_{1,j+1}$
are independent conditional
on $G'_{1,j}$. In another word,
$\{u_j,v_j\}\bot \sigma$.
Thus,
\begin{align}\label{sbmeq19}
\oE_{G'_{1,j+1},\sigma|
G'_{1,j}}
\big[
e^{-\beta I(\sigma(u_j)=\sigma(v_j))}
\big]
=&
\oE_{\sigma|G'_{1,j}}
\bigg[\
\oE_{G'_{1,j+1}|\sigma,G'_{1,j}}
\big[
e^{-\beta I(\sigma(u_j)=\sigma(v_j))}
\big]
\ \bigg]
\\ \nonumber
=&
\oE_{\sigma|G'_{1,j}}
\bigg[\
1+ (e^{-\beta}-1)
\big(\ y(\sigma)^2+(1-y(\sigma))^2\
\big)
\bigg].
\end{align}
By lemma \ref{sbmlem2},
for all $
[\frac{d|\lambda|\delta n }{2(1+r)^2}]+1\leq
j\leq
2[\frac{d|\lambda|\delta n}{2(1+r)^2}]$,
with probability larger than $1-2^{-n}$:
$$\mbbP_{\sigma\sim IS(\beta,G'_{1,j})}\big(\
|y(\sigma)- y^*|> \epsilon_1
\ \big)\leq 2^{-n}.
$$
Therefore continue (\ref{sbmeq19}) we
have that with probability
 $1-e^{-\Omega(n)}$:
for all $[\frac{d|\lambda|\delta n }{2(1+r)^2}]+1\leq
j\leq
2[\frac{d|\lambda|\delta n}{2(1+r)^2}]$,
\begin{align}\label{sbmeq35}
\oE_{G'_{1,j+1},\sigma|
G'_{1,j}}
\big[
e^{-\beta I(\sigma(u_j)=\sigma(v_j))}
\big]
\leq&
\oE_{\sigma|G'_{1,j}}
\big[
\ (e^{-\beta}-1)
\big(\ y(\sigma)^2+(1-y(\sigma))^2\ \big)
\ \big]
\\ \nonumber
\leq&
1+(e^{-\beta}-1)
\big(\
1- 2y^*(1-y^*)
-
2\epsilon_1\ \big)+2^{-n}.
\end{align}
Substituting
(\ref{sbmeq35}) into (\ref{sbmeq36}) and
using inequality $\log(1+x)\leq x$,
the conclusion thus follows.

\section{Concluding remarks}
\label{sbmsec4}

In this paper we evaluate the log partition function
of the Ising model on the SBM with two communities.
The evaluation yields
 a consistent estimator of
the parameter $r$.
We also provid a random clustering algorithm with
positive correlation to the true community label.

\section{Acknowledgement}
We would like to thank Huifeng Peng for helpful
conversation on the topic.
We would like to thank Jing Zhou for comments
on a draft of this work.

\bibliographystyle{apalike}

\bibliography{sbm}

\end{document}